%% file: main.tex
\documentclass[sigconf,screen,nonacm]{acmart}
\pdfoutput=1
\acmConference[LICS]{Logic in Computer Science}{8–12 July, 2024}{Tallinn, Estonia}
\setcopyright{none}

\ccsdesc[500]{Theory of computation~Quantum computation theory}
\ccsdesc[300]{Theory of computation~Linear logic}

\input{preamble.tex}

\usepackage{mathtools}

\pgfsetlayers{background,strings,edgelayer,nodelayer,foreground,main}

\begin{document}

\theoremstyle{acmplain}
\newtheorem*{theorem*}{Theorem}
\theoremstyle{acmplain}
\newtheorem*{lemma*}{Lemma}
\theoremstyle{acmplain}
\newtheorem*{proposition*}{Proposition}

\theoremstyle{acmdefinition}
\newtheorem*{remark}{Remark}

\title{A Profunctorial Semantics for Quantum Supermaps}

\author{James Hefford}
\affiliation{%
	\institution{Université Paris-Saclay, CNRS, ENS Paris-Saclay, Inria, Laboratoire Méthodes Formelles}
	\city{91190, Gif-sur-Yvette}
	\country{France}}
\orcid{0000-0002-6664-8657}
\email{james.hefford@inria.fr}

\author{Matt Wilson}
\affiliation{%
	\institution{Programming Principles Logic and Verification Group, University College London}
	\city{London}
	\country{UK}}
	\orcid{0000-0001-5869-914X}
\email{matthew.wilson.18@ucl.ac.uk}

\keywords{quantum supermaps, higher-order quantum processes, quantum combs, strong profunctors, Tambara modules, coend optics}

\begin{abstract}
We identify morphisms of strong profunctors as a categorification of quantum supermaps.
These black-box generalisations of diagrams-with-holes are hence placed within the broader field of profunctor optics, as morphisms in the category of copresheaves on concrete networks.
This enables the first construction of abstract logical connectives such as tensor products and negations for supermaps in a totally theory-independent setting. 
These logical connectives are found to be all that is needed to abstractly model the key structural features of the quantum theory of supermaps: black-box indefinite causal order, black-box definite causal order, and the factorisation of definitely causally ordered supermaps into concrete circuit diagrams.
We demonstrate that at the heart of these factorisation theorems lies the Yoneda lemma and the notion of representability.
\end{abstract}

\maketitle

\section{Introduction}
Quantum supermaps have been a major focal point in the field of quantum foundations over the last couple of decades \cite{chiribella_supermaps}.
These supermaps are intended to capture the notion of a higher-order quantum process: a first-order process is a quantum channel evolving quantum states in time, while a second-order process is a map which sends first-order processes to first-order processes, that is, quantum channels to quantum channels.

Some simple examples of quantum supermaps include \textit{circuits-with-holes}, also known as \textit{combs} \cite{chiribella_circuits,chiribella_networks}.
These are given by incomplete circuits of quantum channels with holes which one may imagine filling with quantum channels to produce a complete circuit.
One physical conception of such a supermap could be an experimental setup (e.g.\ an arrangement of optical components such as beamsplitters) where one is able to alter some of the components.

Quantum supermaps also encompass substantially more general notions of higher-order transformation some of which have been demonstrated to exhibit exotic phenomena such as superpositions of causal order, and advantage in computational and information-theoretic tasks such as: communication capacity activation \cite{ebler_comms, switch_complex, Chiribella_2021, chiribella_wilson_chau_switch}, channel discrimination \cite{chiribella_switch2, araujo_switch_advantage}, metrology \cite{metrology_switch}, and thermodynamic cooling \cite{fridge_switch}. 
For this reason much of the focus has been on these families of quantum supermaps that go beyond combs.
In fact, the most investigated higher-order processes such as the quantum switch \cite{chiribella_switch}, the OCB process \cite{oreshkov}, the Lugano process \cite{baumeler_incompatibility,baumeler_logically} and the Grenoble process \cite{wechs} are known to possess no decomposition as a comb and thus are truly beyond the class of maps that could be studied in a framework of combs alone.

At their heart, supermaps model a simple intuitive idea: a model of first-order processes consists of boxes and wires, while a model of higher-order processes must extend these compositional components to include holes.
First-order process theories are understood to have a solid mathematical foundation in terms of monoidal category theory \cite{coecke_kissinger_2017}, yet there is not a generally accepted and mathematically rigorous foundation that adequately extends these models to include holes.

The lack of such a foundation is a cause of important domain specific problems.
Firstly, to combine the study of indefinite causal structure with quantum field theory and ultimately quantum gravity, then the supermap framework will need to be extended to infinite dimensional and even non-separable Hilbert spaces \cite{paunkovic2023challenges}.
Current proposals for extension beyond finite dimensions however, are restricted to separable Hilbert spaces, and further to either the $1$-input setting \cite{Chiribella2013NormalCP} or to the Wigner-function representation \cite{Giacomini_2016}.
Secondly, without a stable framework for supermaps in a general context, it is unclear how to initiate the study of post-quantum causal structures: causal structures compatible with GPTs or OPTs \cite{chiribella_purification}, a class of physical theories used to study the special place of quantum physics from information-theoretic principles. 

One way to approach this problem, and the avenue we pursue here, is to develop a general category-theoretic model.
Given that processes, boxes, and wires find formalisation within category theory, it seems natural to expect that boxes and wires with holes should be formalisable on the same terms.
A fully category-theoretic model for supermaps would allow for the study of indefinite causal structures in a variety of potential new domains with quantum theory as a special case.
In these less familiar or more mathematically sophisticated theories, the specification of pre-defined categorical tensor products which capture the main ideas of interest, along with logical rules for how to safely reason with them, would allow for the possibility to reason about supermaps without, for instance, large amounts of domain-specific knowledge.
Consequently, a stable categorical framework has the possibility to open new objects of study in the topic of indefinite causal structure, and furthermore the possibility to provide ways to safely handle those objects.

\subsection{Current Approaches to Quantum Supermaps}\label{sec:current_approaches}
Currently there exist a few approaches to modelling higher order processes categorically.
Perhaps the most successful is that of the $\mathsf{Caus}$-construction \cite{kissinger_caus} which takes a compact closed category $\cat{C}$ of raw processes and equips it with a richer type system in the form of a $*$-autonomous category $\mathsf{Caus}(\cat{C})$ in which one can now ensure that the higher-order processes respect certain notions of causality (such as not creating time loops).
The rich type system, in particular, allows the user to easily consider supermaps on more complex spaces of channels, such as the spaces of non-signalling, or one-way signalling channels. 

One important drawback of the $\mathsf{Caus}$-construction however, is its reliance on compact closure of the starting category $\cat{C}$.
Ultimately this limits its domain of applicability, even within quantum theory where the category of infinite dimensional quantum systems is not compact closed.
Furthermore, it imposes quite a strong condition on $\cat{C}$: since the category is closed, all higher-order processes are already present!
The $\mathsf{Caus}$-construction equips $\cat{C}$ with richer types but it does not generate any new maps, which forces one to make quite strong and concrete choices about the form which supermaps can take from the outset, for instance, linearity is assumed rather than derived in the quantum setting.

A second approach to modelling circuits with holes is that of \textit{coend optics} \cite{pastro_street,clarke_profunctor,roman_coend,roman_comb,roman_optics,roman_thesis,riley_optics,earnshaw,hefford_coend}.
Like the $\mathsf{Caus}$-construction it gives rise to tensor products which allow one to think of the action of circuits with holes on complex spaces, this time constructed by taking categorical limits of other circuit shapes.
Whilst extremely successful in its treatment of circuits with holes, coend optics has a key drawback in the context of the study of quantum supermaps: it does not account for those supermaps of most interest to the physics community, those which \textit{cannot} be decomposed as circuits-with-holes.

A third approach (and the most recently put forward) is that of \cite{wilson_locality,wilson_polycategories} in which a general definition of supermap on plain symmetric monoidal categories was suggested.
This approach does not rely on compact closure, hence being applicable to a much wider range of domains.
Starting from this approach, it was demonstrated that one could recover precisely the quantum supermaps as defined in the field of quantum foundations, including therefore all of those motivating examples which are beyond circuits with holes.

The key observation of \cite{wilson_locality} was an operational one: that a supermap $\eta$ should act locally on its inputs and thus it should commute with the application of combs to any tensorial extension.

\begin{equation}\label{eq:commute_comb}
	\tikzfigscale{1}{figs/supermaplaw1_edit} \hspace{0.8em} = \quad \tikzfigscale{1}{figs/supermaplaw2_edit}
\end{equation}

In \cite{wilson_polycategories} this observation was extended to include multi-party supermaps and a polycategorical semantics were suggested.
This approach is a clear step forward towards modelling supermaps on general symmetric monoidal categories, however, it presently has a number of issues which stem from its lack of identification with established categorical concepts.

Firstly, the definition of \cite{wilson_locality,wilson_polycategories} is quite intractable algebraically, relying instead on unformalised diagrammatic notations along with assumptions on the strictness of the category at hand.
For the purposes of quantum foundations research, perhaps this is not so problematic, but for the purposes of putting together a general categorical model for supermaps, it leaves the approach clearly still in a primordial stage. 

Secondly, and related to the first, this approach does not come with the same kinds of logical connectives such as tensor products which the previous two more categorically grounded approaches provide \cite{riley_optics,kissinger_caus, pastro_street}. It is even unclear how to identify spaces within this concrete approach on which the multi-input supermaps can be interpreted as acting (leaving us without an analogy to the pleasing connection between supermaps and the signalling constraints on their domains \cite{chiribella_switch, Bisio_2019, kissinger_caus}), or how even how to reformulate the axiomatic hierarchy of supermaps with definite causal order \cite{chiribella_networks}. 

\subsection{Contributions}
In this paper the goal is to produce a categorical model for supermaps which is both structurally well-behaved (with, for instance, suitable tensor products) but still fully general in the categories it can be applied to and the kinds of higher order processes incorporated.
The key idea is to make a direct connection between the methods of profunctor optics \cite{pastro_street,clarke_profunctor,roman_coend,roman_comb,roman_optics,roman_thesis,riley_optics,earnshaw,hefford_coend} and the diagrammatic approach of \cite{wilson_locality,wilson_polycategories}, by identifying the aforementioned diagrammatic rules as an interpretation of the definition of morphism of strong profunctors.
By making this identification, and defining supermaps using profunctorial methods, it is possible to go beyond what could be concretely imagined in diagrammatic approaches, and in dealing with these more abstract spaces we are able to solve some of the key issues in providing a good categorical model for supermaps.

It is known that the category $\StProf(\cat{C})$ has two closed tensor products \cite{pastro_street,garner,earnshaw} which we will demonstrate model the spaces of separable ($\otimes_\cat{C}$) and sequenced channels ($\seq$).
As a result the morphisms with those domains are shown to faithfully represent the supermaps with indefinite causal order and definite causal order respectively.
Moreover, these tensors interact to make $\StProf(\cat{C})$ a normal closed duoidal category \cite{garner,earnshaw}.
Such categories have been identified in \cite{shapiro_duoidal} as a way of modelling notions of dependent and independent composition.
Indeed, here we will demonstrate that they allow us to model parallel and sequential composition of families of maps, and thereby capture separable and semi-localisable channels.
We can think of the tensors $\otimes_\cat{C}$ and $\seq$ as analogous to the tensor $\otimes$ and sequencing operator $\oslash$ of the $\caus$-construction: their normal duoidal interaction is a fragment of the BV structure suggested in \cite{SimmonsKissinger2022}.

In short, we are able to model the three core concepts of supermaps in the literature all within a single category $\StProf(\cat{C})$.

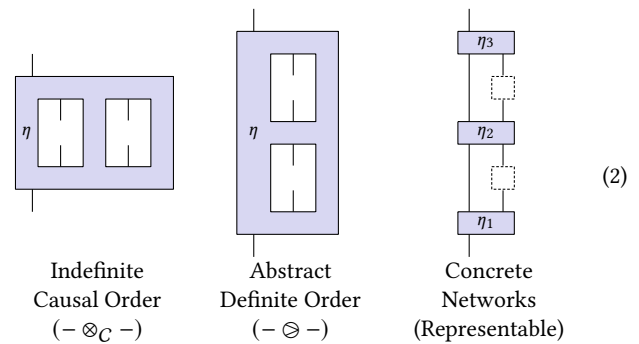
\begin{equation}\label{eq:different_supermaps}
	\begin{tabular}{ccc}
		\tikzfigscale{1}{figs/horizontal_physics_1} & \tikzfigscale{1}{figs/vertical_physics_1} & \tikzfigscale{1}{figs/vertical_physics_3} \\
		\begin{tabular}[t]{c} Indefinite \\ Causal Order \\ $( - \otimes_{\cat{C}} - )$ \end{tabular} & 
		\begin{tabular}[t]{c} Abstract \\ Definite Order \\ $( -\seq -)$ \end{tabular} & 
		\begin{tabular}[t]{c} Concrete \\ Networks \\ (Representable) \end{tabular}
	\end{tabular}
\end{equation}

$\StProf(\cat{C})$ also has a weak dualising functor $(-)^*$, which in general is not involutive but is still strong with respect to the tensor $\otimes_\cat{C}$.
\begin{align*}
	(-)^* = [-,i] \qquad \text{weak duals} \\
	- \amp - \qquad \text{all processes}
\end{align*}
This allows us to define a functorial par operation $\amp$ which has many of the same properties as its analogous one in $\caus$ (for instance it distributes linearly over $\otimes_\cat{C}$) but it is generally \textit{not} associative or unital and thus fails to be a tensor product.
This weakens the $*$-autonomous structure of the $\mathsf{Caus}$-construction, meaning we do not have a model of linear logic but of \textit{tensorial logic} \cite{mellies_tensorlogic}.
The failure of $(-)^*$ to be involutive has a number of deep connections with the structure of supermaps on the category $\cat{C}$ and we will demonstrate that the involutivity of $(-)^*$ on certain objects of $\StProf(\cat{C})$ is intimately connected with the fundamental decomposition theorems of quantum supermaps, which demonstrate that abstract definite orders in quantum theory are realisable as concrete networks.

\section{Single-Party Supermaps as Strong Natural Transformations}
Let us start with the most simple case of supermaps on a single party $\vb*{a}$.
By a \textit{party} $\vb*{a}$ we mean a pair $(a,a')$ of objects of a base symmetric monoidal category $\cat{C}$.
We think of $a$ as the input and $a'$ as the output of interventions for a lab in a spacetime at $\vb*{a}$.
A supermap sending $\vb*{a}\morph{}\vb*{b}$, as conceptualised in the quantum foundations literature \cite{chiribella_supermaps}, consists of a transformation sending the lab at $\vb*{a}$ to the lab at $\vb*{b}$ and is often drawn informally in the following fashion.

\begin{equation}\label{eq:normal_supermap}
	\tikzfigscale{1}{figs/normal_supermap_1} \ :: \  \tikzfigscale{1}{figs/normal_supermap_2a} \ \mapsto \   \tikzfigscale{1}{figs/normal_supermap_2}
\end{equation}

Such a supermap takes channels $a\morph{}a'$ representing interventions at $\vb*{a}$ to channels $b\morph{}b'$: it is like a hole in a circuit, but with no a priori assumption that the map decomposes as a circuit of maps from $\cat{C}$.

For arbitrary symmetric monoidal categories, it is not clear how to rigorously define $\eta$, by either the non-categorical methods in the quantum foundations literature, or by the $\caus$-construction.
In all those cases, additional structure is assumed of $\cat{C}$, typically compact closure, such that one may bend around wires of $\eta$ and interpret it as a morphism in the base category $\cat{C}$ \cite{kissinger_caus}.

In \cite{wilson_locality} a general method was outlined, though it did not explain its connection to more well-known categorical structures. 
Our aim of this section is to unpack this definition and make the connection with profunctorial methods apparent.

\begin{definition}\label{defn:lat}
	A \textit{locally-applicable transformation} $\eta:\vb*{a}\morph{}\vb*{b}$ consists of a family of functions of the form:
	\begin{equation*}
		\tikzfigscale{1}{figs/supermap_family}
	\end{equation*}
	indexed by the pair of objects $(x,x')$, which assign to each morphism $a\otimes x\morph{} a'\otimes x'$ a morphism $b\otimes x\morph{} b'\otimes x'$, such that the following laws hold.
	\begin{equation}\label{eq:nat_law}
		\tikzfigscale{1}{figs/supermap_nat1} \quad = \quad \tikzfigscale{1}{figs/supermap_nat2}
	\end{equation}
	\begin{equation}\label{eq:strength_law}
		\tikzfigscale{1}{figs/supermap_strength1} \quad = \quad \tikzfigscale{1}{figs/supermap_strength2}
	\end{equation}
\end{definition}

The idea here is to capture the \textit{locality} of the supermap $\eta$: that it should only act locally to the lab $\vb*{a}$ and not on other environment systems $\vb*{x}=(x,x')$.
Thus we expect a supermap to commute with the actions of agents on $(x,x')$, in particular with pre- and post-composition by maps on $(x,x')$ (law \eqref{eq:nat_law}) and by further tensorial extension of the environment (law \eqref{eq:strength_law}).
These two laws jointly capture the intuitive idea that supermaps should commute with combs on the environment, see \eqref{eq:commute_comb}.

To formalise Definition \ref{defn:lat}, such that locally-applicable transformations form a category themselves, we must take care to define both the objects that represent the parties $\vb*{a}$, and the arrows that will represent the supermaps.

\subsection{Basic Types for (Co)Domains of Supermaps}
Let us start with formalising the objects that represent the parties $\vb*{a}$.
The first thing to note is that each function $\eta_{x,x'}$ takes the space of bipartite channels of type $a\otimes x\morph{} a'\otimes x'$ to the space of bipartite channels of type $b\otimes x\morph{} b'\otimes x'$.
These spaces of channels have a nice categorical formalisation.
We write $\cat{C}(a,a')$ for the set of morphisms $a\morph{}a'$ of $\cat{C}$.
The family of such sets, varying over $(a,a')$, has two key properties: they behave well with respect to composition with morphisms of $\cat{C}$ and they behave well with respect to the tensor product of $\cat{C}$.

\subsubsection{Compatibility with Composition}\label{sec:composition}
Given a morphism $g:a'\morph{}b'$ there is a function which acts by post-composition with $g$.
\begin{equation*}
	\cat{C}(a,a') \morph{\cat{C}(1,g)} \cat{C}(a,b') :: (a\morph{\phi}a') \mapsto (a\morph{\phi}a'\morph{g}b')
\end{equation*}
Furthermore, given a morphism $f:b\morph{}a$ there is a function which acts by pre-composition with $f$.
\begin{equation*}
	\cat{C}(a,a') \morph{\cat{C}(f,1)} \cat{C}(b,a') :: (a\morph{\phi}a') \mapsto (b\morph{f}a\morph{\phi}a')
\end{equation*}
These functions are in fact functorial over the category $\cat{C}$ which makes the whole family $\cat{C}(-,\bl)$ a special instance of a \textit{profunctor}, known in this case as the hom-profunctor.

\begin{definition}
	A \textit{profunctor} $P:\cat{C}\pmorph\cat{D}$ is a functor $P:\opcat{\cat{D}}\times\cat{C}\morph{}\set$.
\end{definition}
Profunctors provide a categorical generalisation of several other familiar mathematical objects including relations between sets, distributions and bimodules.
Indeed, they are also known as relators \cite{loregian_coend}, distributors \cite{benabou} and often simply as bimodules.
This latter conception is a valuable one to keep in mind during this article.
Much as how a bimodule is an abelian group with commuting left and right actions by two rings, a profunctor $P$ consists of a family of sets with commuting, functorial left and right actions by two categories.
\begin{equation*}
	P(d,c) \morph{P(1,g)} P(d,c') :: (d\overset{p}{\rightsquigarrow} c) \mapsto (d\overset{p}{\rightsquigarrow}c\morph{g}c')
\end{equation*}
\begin{equation*}
	P(d,c) \morph{P(f,1)} P(d',c) :: (d\overset{p}{\rightsquigarrow}c) \mapsto (d'\morph{f}d\overset{p}{\rightsquigarrow}c)
\end{equation*}
Above we have used some informal notation purely to convey a certain intuition that will be valuable in this work, more formally by $p:d\rightsquigarrow c$ we mean $p\in P(d,c)$ and the ``compositions'' are given by function evaluation e.g.\ $P(1,g)(p)$.

Profunctors can be composed by taking a \textit{coend}.
\begin{equation*}
  (Q\circ P)(-,\bl) := \int^d Q(-,d)\times P(d,\bl),
\end{equation*}
The coend provides a generalisation of the tensor product of bimodules and acts to quotient the left action on $P$ with the right action on $Q$.
Coends can be formalised in terms of universal cowedges and can be formally written as a certain colimit; for the details of this we refer the reader to \cite{loregian_coend}.
For us, we can think of the coend as having the effect of quotienting to make the following two compositions equivalent for any $f$.
\begin{equation*}
	(e\overset{q}{\rightsquigarrow}d\morph{f}d') \overset{p}{\rightsquigarrow} c \quad \sim \quad e\overset{q}{\rightsquigarrow} (d\morph{f}d' \overset{p}{\rightsquigarrow} c )
\end{equation*}

The hom-profunctors $1_{\cat{C}}=\cat{C}(-,\bl):\cat{C}\pmorph\cat{C}$ are the identity profunctors for this composition due to the following instances of the Yoneda lemma.
\begin{equation*}
  \int^d \cat{D}(-,d)\times P(d,\bl) \cong P(-,\bl) \cong \int^c P(-,c)\times \cat{C}(c,\bl).
\end{equation*}
One way of conceptualising these identities is to think of the hom-profunctors as behaving a bit like Dirac delta functions.

\subsubsection{Compatibility with the Tensor}\label{sec:tensor}

The issue for us with plain profunctors is that they have no inherent compatibility with the monoidal structure of $\cat{C}$.
To resolve this we will be interested in studying the strong profunctors---those that behave well with the tensor product.
Strong profunctors are also known as Tambara modules \cite{tambara,pastro_street}.

\begin{definition}
  Let $(\cat{C},\otimes,i)$ be a symmetric monoidal category.
  A profunctor $P:\cat{C}\pmorph\cat{C}$ is \textit{strong} when it is equipped with a natural transformation,
  \begin{equation*}
    P(a,b)\times\cat{C}(c,d) \morph{\zeta_{abcd}} P(a\otimes c, b\otimes d),
  \end{equation*}
  such that certain diagrams commute making this action essentially associative and unital \cite{tambara,pastro_street}.
\end{definition}

While profunctors are bimodular in the direction of composition in $\cat{C}$ (the ``vertical'' direction in our diagrams), strong profunctors augment this with an additional bimodular nature in the ``horizontal'' tensorial direction.
We can think of the component $\zeta_{abcd}$ as having the effect of ``tensoring'' the elements $p:a\rightsquigarrow b$ of $P(a,b)$ with morphisms of $\cat{C}$.
\begin{equation*}
	(a\overset{p}{\rightsquigarrow} b, c\morph{f}d) \mapsto (a\otimes c \overset{p\cdot f}{\rightsquigarrow} b\otimes d)
\end{equation*}

A key example of a strong profunctor is the hom-profunctor.
Its strength is simply given by taking tensor products in $\cat{C}$.
\begin{equation*}
	\cat{C}(a,b)\times\cat{C}(c,d) \morph{} \cat{C}(a\otimes c,b\otimes d) :: (f,g) \mapsto f\otimes g
\end{equation*}

\subsubsection{Basic Supermap Types}
Combining the discussion of Sections \ref{sec:composition} and \ref{sec:tensor} we can see that the domain and codomain of the supermap $\eta$ can be taken to be of the form $\suphom{a}{a'}$.
This constitutes the space of all bipartite inputs $a\otimes x\morph{}a'\otimes x'$ to the supermap where $a$ and $a'$ are fixed, while $x$ and $x'$ are allowed to vary.
$\suphom{a}{a'}$ is a strong endoprofunctor on $\cat{C}$, with its functoriality and strength essentially identical to those outlined for the hom-profunctor in the previous two sub-sections.
Explicitly, the strength is given as before by taking tensor products in $\cat{C}$.
\begin{align*}
  \cat{C}(a\otimes x,a'\otimes x')\times\cat{C}(y,y') \morph{\zeta_{xx'yy'}} \cat{C}(a\otimes x\otimes y, a'\otimes x'\otimes y')\\
  :: (f,g)\mapsto f\otimes g
\end{align*}

\subsection{Supermaps as Arrows on Strong Profunctors}
Having established the basic properties of the strong endoprofunctors $\suphom{a}{a'}$ which will form the (co)domains of single-party supermaps, we now turn our attention to the morphisms between strong profunctors.

Plain endoprofunctors on a category $\cat{C}$ assemble to give a category $\Prof(\cat{C})$, where the morphisms are the natural transformations.
Strong endoprofunctors on a monoidal category also assemble into a category $\StProf(\cat{C})$ where the morphisms are the \textit{strong} natural transformations.

\begin{definition}
	A natural transformation $\eta:P\Rightarrow Q$ between strong endoprofunctors is itself \textit{strong} when it commutes with the strengths.
	\begin{equation}\label{eq:strong_nat}
		\begin{tikzcd}
		P(a,b)\times \cat{C}(c,d) \ar[d,"\eta_{ab}\times 1"'] \ar[r,"\zeta_{abcd}^P"] & P(a\otimes c,b\otimes d) \ar[d,"\eta_{a\otimes c,b\otimes d}"] \\
		Q(a,b) \ar[r,"\zeta_{abcd}^Q"'] \times \cat{C}(c,d) & Q(a\otimes c,b\otimes d)
		\end{tikzcd}
	\end{equation}
\end{definition}

Thinking of a strength as a way of tensoring elements of a profunctor by morphisms of $\cat{C}$, a natural transformation is strong precisely when it preserves this structure, $\eta(p\cdot f) = \eta(p)\cdot f$.

Consider now a strong natural transformation $\eta: \suphom{a}{a'} \Rightarrow \suphom{b}{b'}$.
To ask that $\eta$ is natural is to ask that the following diagram commutes.
\begin{equation*}
	\begin{tikzcd}[column sep=large]
		\cat{C}(a\otimes x,a'\otimes x') \ar[r,"{\cat{C}(1\otimes f,1\otimes g)}"] \ar[d,"{\eta_{xx'}}"'] & \cat{C}(a\otimes y, a'\otimes y') \ar[d,"{\eta_{yy'}}"] \\
		\cat{C}(b\otimes x,b'\otimes x') \ar[r,"{\cat{C}(1\otimes f,1\otimes g)}"'] & \cat{C}(b\otimes y, b'\otimes y')
	\end{tikzcd}
\end{equation*}
This is precisely law \eqref{eq:nat_law}.
To ask that $\eta$ is strong is to ask $\eta(\phi \otimes f) = \eta(\phi) \otimes f$ (upon unpacking diagram \eqref{eq:strong_nat}).
This is precisely law \eqref{eq:strength_law}.

This means we can restate the definition of a locally-applicable transformation on single-parties more formally as follows.

\begin{definition}\label{defn:basic_supermap}
	A single-party locally-applicable transformation is a strong natural transformation of the type
	\begin{equation}\label{eq:basic_supermap}
		\eta: \suphom{a}{a'} \Rightarrow \suphom{b}{b'}.
	\end{equation}
\end{definition}

At this point it is perhaps instructive to see some simple examples of single-party locally applicable transformations.
To do this we will employ the \textit{internal string diagrams} of \cite{bartlett_extended,bartlett_modular}.
These consist of the usual string diagrams for monoidal categories, bounded inside cobordisms.
The cobordisms or ``tubes'' represent certain profunctors, while the string diagrams inside them represent elements of those profunctors at some chosen objects of the categories involved.
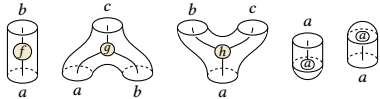
\begin{equation*}
  \begin{tikzpicture}[baseline=(current bounding box.center),{every node/.style}={scale=0.8}]
    \node[tube,top,bot] (t) {};
    \begin{pgfonlayer}{nodelayer}
      \node[label] (f) at (t.center) {$f$};
	  \node[system] at ([yshift=0.6\ylength]t.top) {$b$};
	  \node[system] at ([yshift=-0.6\ylength]t.bot) {$a$};
    \end{pgfonlayer}
    \begin{pgfonlayer}{strings}
      \draw (t.bot) to (f) to (t.top) {};
    \end{pgfonlayer}
  \end{tikzpicture}\quad
  \begin{tikzpicture}[baseline=(current bounding box.center),{every node/.style}={scale=0.8}]
    \node[pants,top,bot] (p) {};
    \begin{pgfonlayer}{nodelayer}
      \node[label] (g) at (p.center) {$g$};
	  \node[system] at ([yshift=0.6\ylength]p.belt) {$c$};
	  \node[system] at ([yshift=-0.6\ylength]p.leftleg) {$a$};
	  \node[system] at ([yshift=-0.6\ylength]p.rightleg) {$b$};
    \end{pgfonlayer}
    \begin{pgfonlayer}{strings}
      \draw (g) to (p.belt) {};
      \draw (g) to[out=-155,in=90] (p.leftleg) {};
      \draw (g) to[out=-25,in=90] (p.rightleg) {};
    \end{pgfonlayer}
  \end{tikzpicture}\quad
  \begin{tikzpicture}[baseline=(current bounding box.center),{every node/.style}={scale=0.8}]
    \node[copants,top,bot] (c) {};
    \begin{pgfonlayer}{nodelayer}
      \node[label] (h) at (c.center) {$h$};
	  \node[system] at ([yshift=-0.6\ylength]c.belt) {$a$};
	  \node[system] at ([yshift=0.6\ylength]c.leftleg) {$b$};
	  \node[system] at ([yshift=0.6\ylength]c.rightleg) {$c$};
    \end{pgfonlayer}
    \begin{pgfonlayer}{strings}
      \draw (h) to (c.belt) {};
      \draw (h) to[out=155,in=-90] (c.leftleg) {};
      \draw (h) to[out=25,in=-90] (c.rightleg) {};
    \end{pgfonlayer}
  \end{tikzpicture}\quad
  \begin{tikzpicture}[baseline=(current bounding box.center),{every node/.style}={scale=0.8}]
    \node[bowl] (b) {};
    \node[tube,anchor=bot,yscale=0.5] (t) at (b) {};
    \node[end] at (t.top) {};
    \node[end_dot] at (t.bot) {};
    \begin{pgfonlayer}{nodelayer}
      \node[representable] (a) at (b.center) {$a$};
	  \node[system] at ([yshift=0.6\ylength]t.top) {$a$};
    \end{pgfonlayer}
    \begin{pgfonlayer}{strings}
      \draw (a) to (t.top) {};
    \end{pgfonlayer}
  \end{tikzpicture}\quad
  \begin{tikzpicture}[baseline=(current bounding box.center),{every node/.style}={scale=0.8}]
    \node[cobowl] (b) {};
    \node[tube,anchor=top,yscale=0.5] (t) at (b) {};
    \node[end_dot] at (t.bot) {};
    \begin{pgfonlayer}{nodelayer}
      \node[representable] (a) at (b.center) {$a$};
	  \node[system] at ([yshift=-0.6\ylength]t.bot) {$a$};
    \end{pgfonlayer}
    \begin{pgfonlayer}{strings}
      \draw (t.bot) to (a) {};
    \end{pgfonlayer}
  \end{tikzpicture}
\end{equation*}
For instance, in the left-most diagram above, the tube is the profunctor $\cat{C}(-,\bl):\cat{C}\pmorph\cat{C}$ and the internal string diagram is the morphism $f\in\cat{C}(a,b)$ at the chosen objects $a$ and $b$.
Altogether this is a morphism $(f,\cat{C}(-,\bl)):(\cat{C},b)\pmorph(\cat{C},a)$.
The other diagrams represent, from left to right, $(g,\cat{C}(-\otimes -,\bl)):(\cat{C},c)\pmorph(\cat{C}\times\cat{C},(a,b))$, $(h,\cat{C}(-,\bl \otimes \bl):(\cat{C}\times\cat{C},(b,c))\pmorph(\cat{C},a)$, $(1_a,\cat{C}(a,-)):(\cat{C},a)\pmorph(1,*)$, and $(1_a,\cat{C}(-,a)):(1,*)\pmorph(\cat{C},a)$, where $1$ is the terminal category with single object $*$.

A formal semantics can be given to these diagrams in a number of ways, as an instance of the \textit{wire diagrams} of \cite{bartlett_wirediagrams} for the monoidal bicategory of pointed profunctors \cite{roman_coend} or as diagrams for the tricategory of pointed strong profunctors \cite{braithwaite_collages}.

\begin{xlrbox}{supermap_basic1}
	\begin{tikzpicture}[baseline=(current bounding box.center),{every node/.style}={scale=0.8}]
		\node[copants,bot,widebelt] (c) {};
		\node[pants,bot,widebelt,anchor=belt] (p) at (c.belt) {};
		\node[bowl] at (p.leftleg) {};
		\node[cobowl] at (c.leftleg) {};
		\node[end] at (c.rightleg) {};
		\begin{pgfonlayer}{nodelayer}
			\node[label] (m) at (p.belt) {$\phi$};
			\node[representable] at (c.leftleg) {$a'$};
			\node[representable] at (p.leftleg) {$a$};
			\node[system] at ([yshift=0.6\ylength]c.rightleg) {$x'$};
			\node[system] at ([yshift=-0.6\ylength]p.rightleg) {$x$};
		\end{pgfonlayer}
		\begin{pgfonlayer}{strings}
			\draw (m) to[out=-135,in=90] (p.leftleg) {};
			\draw (m) to[out=-45,in=90] (p.rightleg) {};
			\draw (m) to[out=135,in=-90] (c.leftleg) {};
			\draw (m) to[out=45,in=-90] (c.rightleg) {};
		\end{pgfonlayer}
	\end{tikzpicture}
\end{xlrbox}

\begin{xlrbox}{supermap_basic2}
	\begin{tikzpicture}[baseline=(current bounding box.center),{every node/.style}={scale=0.8}]
		\node[copants,bot,widebelt] (c) {};
		\node[pants,bot,widebelt,anchor=belt] (p) at (c.belt) {};
		\node[bowl] at (p.leftleg) {};
		\node[cobowl] at (c.leftleg) {};
		\node[end] at (c.rightleg) {};
		\begin{pgfonlayer}{nodelayer}
			\node[label] (m) at (p.belt) {$\eta\phi$};
			\node[representable] at (c.leftleg) {$b'$};
			\node[representable] at (p.leftleg) {$b$};
			\node[system] at ([yshift=0.6\ylength]c.rightleg) {$x'$};
			\node[system] at ([yshift=-0.6\ylength]p.rightleg) {$x$};
		\end{pgfonlayer} 
		\begin{pgfonlayer}{strings}
			\draw (m) to[out=-135,in=90] (p.leftleg) {};
			\draw (m) to[out=-45,in=90] (p.rightleg) {};
			\draw (m) to[out=135,in=-90] (c.leftleg) {};
			\draw (m) to[out=45,in=-90] (c.rightleg) {};
		\end{pgfonlayer}
	\end{tikzpicture}
\end{xlrbox}

In the internal string diagrams a single-party locally-applicable transformation is given by an arrow of the following type,
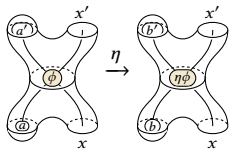
\begin{equation*}
	\xusebox{supermap_basic1} \morph{\eta} \xusebox{supermap_basic2}
\end{equation*}
that is compatible with the strengths,

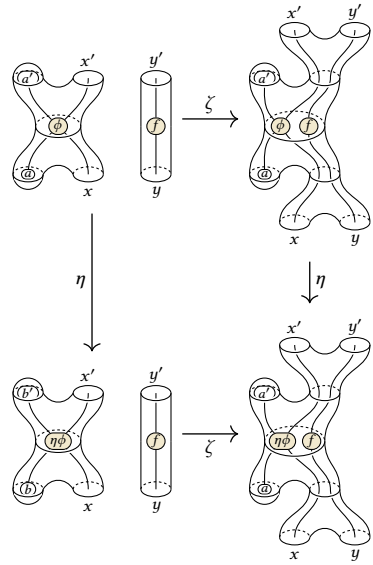
\begin{xlrbox}{supermap_basic_str1}
	\begin{tikzpicture}[baseline=(current bounding box.center),{every node/.style}={scale=0.8}]
		\node[copants,bot,widebelt] (c) {};
		\node[pants,bot,widebelt,anchor=belt] (p) at (c.belt) {};
		\node[bowl] at (p.leftleg) {};
		\node[cobowl] at (c.leftleg) {};
		\node[end] at (c.rightleg) {};
		\node[tube,top,anchor=top] (t) at ([xshift=1.8\xlength]c.rightleg) {};
		\node[tube,bot,anchor=top] (t2) at (t.bot) {};
		\begin{pgfonlayer}{nodelayer}
			\node[label] (m) at (p.belt) {$\phi$};
			\node[representable] at (c.leftleg) {$a'$};
			\node[representable] at (p.leftleg) {$a$};
			\node[label] (f) at (t.bot) {$f$};
			\node[system] at ([yshift=0.6\ylength]c.rightleg) {$x'$};
			\node[system] at ([yshift=-0.6\ylength]p.rightleg) {$x$};
			\node[system] at ([yshift=0.6\ylength]t.top) {$y'$};
			\node[system] at ([yshift=-0.6\ylength]t2.bot) {$y$};
		\end{pgfonlayer}
		\begin{pgfonlayer}{strings}
			\draw (m) to[out=-135,in=90] (p.leftleg) {};
			\draw (m) to[out=-45,in=90] (p.rightleg) {};
			\draw (m) to[out=135,in=-90] (c.leftleg) {};
			\draw (m) to[out=45,in=-90] (c.rightleg) {};
			\draw (t2.bot) to (f) to (t.top) {};
		\end{pgfonlayer}
	\end{tikzpicture}
\end{xlrbox}

\begin{xlrbox}{supermap_basic_str2}
	\begin{tikzpicture}[baseline=(current bounding box.center),{every node/.style}={scale=0.8}]
		\node[copants,bot,widerbelt] (c) {};
		\node[pants,bot,widerbelt,anchor=belt] (p) at (c.belt) {};
		\node[bowl] at (p.leftleg) {};
		\node[cobowl] at (c.leftleg) {};
		\node[copants,top,bot,anchor=belt] (c2) at (c.rightleg) {};
		\node[pants,bot,anchor=belt] (p2) at (p.rightleg) {};
		\begin{pgfonlayer}{nodelayer}
			\node[label] (m) at ([xshift=-.75\tubewidth]p.belt) {$\phi$};
			\node[label] (f) at ([xshift=.75\tubewidth]p.belt) {$f$};
			\node[representable] at (c.leftleg) {$a'$};
			\node[representable] at (p.leftleg) {$a$};
			\node[system] at ([yshift=0.6\ylength]c2.leftleg) {$x'$};
			\node[system] at ([yshift=-0.6\ylength]p2.leftleg) {$x$};
			\node[system] at ([yshift=0.6\ylength]c2.rightleg) {$y'$};
			\node[system] at ([yshift=-0.6\ylength]p2.rightleg) {$y$};
		\end{pgfonlayer}
		\begin{pgfonlayer}{strings}
			\draw (m) to[out=-135,in=90] (p.leftleg) {};
			\draw (m) to[out=-55,in=90] ([xshift=-.3\tubewidth]p.rightleg) to[out=-90,in=90] (p2.leftleg) {};
			\draw (m) to[out=135,in=-90] (c.leftleg) {};
			\draw (m) to[out=55,in=-90] ([xshift=-.3\tubewidth]c.rightleg) to[out=90,in=-90] (c2.leftleg) {};
			\draw (f) to[out=90,in=-90] ([xshift=.3\tubewidth]c.rightleg) to[out=90,in=-90] (c2.rightleg) {};
			\draw (f) to[out=-90,in=90] ([xshift=.3\tubewidth]p.rightleg) to[out=-90,in=90] (p2.rightleg) {};
		\end{pgfonlayer}
	\end{tikzpicture}
\end{xlrbox}

\begin{xlrbox}{supermap_basic_str3}
	\begin{tikzpicture}[baseline=(current bounding box.center),{every node/.style}={scale=0.8}]
		\node[copants,bot,widebelt] (c) {};
		\node[pants,bot,widebelt,anchor=belt] (p) at (c.belt) {};
		\node[bowl] at (p.leftleg) {};
		\node[cobowl] at (c.leftleg) {};
		\node[end] at (c.rightleg) {};
		\node[tube,top,anchor=top] (t) at ([xshift=1.8\xlength]c.rightleg) {};
		\node[tube,bot,anchor=top] (t2) at (t.bot) {};
		\begin{pgfonlayer}{nodelayer}
			\node[label] (m) at (p.belt) {$\eta\phi$};
			\node[representable] at (c.leftleg) {$b'$};
			\node[representable] at (p.leftleg) {$b$};
			\node[label] (f) at (t.bot) {$f$};
			\node[system] at ([yshift=0.6\ylength]c.rightleg) {$x'$};
			\node[system] at ([yshift=-0.6\ylength]p.rightleg) {$x$};
			\node[system] at ([yshift=0.6\ylength]t.top) {$y'$};
			\node[system] at ([yshift=-0.6\ylength]t2.bot) {$y$};
		\end{pgfonlayer}
		\begin{pgfonlayer}{strings}
			\draw (m) to[out=-135,in=90] (p.leftleg) {};
			\draw (m) to[out=-45,in=90] (p.rightleg) {};
			\draw (m) to[out=135,in=-90] (c.leftleg) {};
			\draw (m) to[out=45,in=-90] (c.rightleg) {};
			\draw (t2.bot) to (f) to (t.top) {};
		\end{pgfonlayer}
	\end{tikzpicture}
\end{xlrbox}

\begin{xlrbox}{supermap_basic_str4}
	\begin{tikzpicture}[baseline=(current bounding box.center),{every node/.style}={scale=0.8}]
		\node[copants,bot,widerbelt] (c) {};
		\node[pants,bot,widerbelt,anchor=belt] (p) at (c.belt) {};
		\node[bowl] at (p.leftleg) {};
		\node[cobowl] at (c.leftleg) {};
		\node[copants,top,bot,anchor=belt] (c2) at (c.rightleg) {};
		\node[pants,bot,anchor=belt] (p2) at (p.rightleg) {};
		\begin{pgfonlayer}{nodelayer}
			\node[label] (m) at ([xshift=-.65\tubewidth]p.belt) {$\eta\phi$};
			\node[label] (f) at ([xshift=0.9\tubewidth]p.belt) {$f$};
			\node[representable] at (c.leftleg) {$a'$};
			\node[representable] at (p.leftleg) {$a$};
			\node[system] at ([yshift=0.6\ylength]c2.leftleg) {$x'$};
			\node[system] at ([yshift=-0.6\ylength]p2.leftleg) {$x$};
			\node[system] at ([yshift=0.6\ylength]c2.rightleg) {$y'$};
			\node[system] at ([yshift=-0.6\ylength]p2.rightleg) {$y$};
		\end{pgfonlayer}
		\begin{pgfonlayer}{strings}
			\draw (m) to[out=-135,in=90] (p.leftleg) {};
			\draw (m) to[out=-55,in=90] ([xshift=-.3\tubewidth]p.rightleg) to[out=-90,in=90] (p2.leftleg) {};
			\draw (m) to[out=135,in=-90] (c.leftleg) {};
			\draw (m) to[out=55,in=-90] ([xshift=-.3\tubewidth]c.rightleg) to[out=90,in=-90] (c2.leftleg) {};
			\draw (f) to[out=90,in=-90] ([xshift=.3\tubewidth]c.rightleg) to[out=90,in=-90] (c2.rightleg) {};
			\draw (f) to[out=-90,in=90] ([xshift=.3\tubewidth]p.rightleg) to[out=-90,in=90] (p2.rightleg) {};
		\end{pgfonlayer}
	\end{tikzpicture}
\end{xlrbox}

\begin{equation*}
	\begin{tikzcd}
		\xusebox{supermap_basic_str1} \ar[r,"\zeta"] \ar[d,"\eta"'] & \xusebox{supermap_basic_str2} \ar[d,"\eta"] \\
		\xusebox{supermap_basic_str3} \ar[r,"\zeta"'] & \xusebox{supermap_basic_str4}
	\end{tikzcd}
\end{equation*}

\begin{xlrbox}{supermap_comb1}
	\begin{tikzpicture}[baseline=(current bounding box.center),{every node/.style}={scale=0.8}]
	  \node[copants,bot,widebelt] (c) {};
	  \node[pants,bot,widebelt,anchor=belt] (p) at (c.belt) {};
	  \node[bowl] at (p.leftleg) {};
	  \node[cobowl] at (c.leftleg) {};
	  \node[end] at (c.rightleg) {};
	  \node[bowl] (b) at ([yshift=1.25\ylength]c.leftleg) {};
	  \node[pants,bot,anchor=rightleg] (p1) at (b) {};
	  \node[cobowl] at (p1.belt) {};
	  \node[tube,anchor=top,yscale=3.5] (t1) at (p1.leftleg) {};
	  \node[copants,bot,anchor=leftleg] (c1) at (t1.bot) {};
	  \node[cobowl] at (c1.rightleg) {};
	  \node[bowl] at (c1.belt) {};
	  \node[end_dot] at (c1.leftleg) {};
	  \begin{pgfonlayer}{nodelayer}
	    \node[label] (m) at (p.belt) {$\phi$};
	    \node[representable] at (c.leftleg) {$a'$};
	    \node[representable] at (p.leftleg) {$a$};
	    \node[system] at ([yshift=0.6\ylength]c.rightleg) {$x'$};
	    \node[system] at ([yshift=-0.6\ylength]p.rightleg) {$x$};
	    \node[label] (f) at (p1.center) {$g$};
	    \node[label] (g) at (c1.center) {$f$};
	    \node[representable] (b) at (c1.belt) {$b$};
	    \node[representable] (a) at (c1.rightleg) {$a$};
	    \node[representable] (a') at (p1.rightleg) {$a'$};
	    \node[representable] (b') at (p1.belt) {$b'$};
	  \end{pgfonlayer}
	  \begin{pgfonlayer}{strings}
	    \draw (m) to[out=-135,in=90] (p.leftleg) {};
	    \draw (m) to[out=-45,in=90] (p.rightleg) {};
	    \draw (m) to[out=135,in=-90] (c.leftleg) {};
	    \draw (m) to[out=45,in=-90] (c.rightleg) {};
	    \draw (f) to (b') {};
	    \draw (f) to[out=-25,in=90] (p1.rightleg) {};
	    \draw (f) to[out=-155,in=90] (p1.leftleg) to (c1.leftleg) to[out=-90,in=155] (g) {};
	    \draw (g) to[out=25,in=-90] (c1.rightleg) {};
	    \draw (g) to (b) {};
	  \end{pgfonlayer}
	\end{tikzpicture}
\end{xlrbox}

\begin{xlrbox}{supermap_comb2}
   \begin{tikzpicture}[baseline=(current bounding box.center),{every node/.style}={scale=0.8}]
	  \node[copants,bot,widebelt] (c) {};
	  \node[pants,bot,widebelt,anchor=belt] (p) at (c.belt) {};
	  \node[end] at (c.rightleg) {};
	  \node[pants,bot,anchor=rightleg] (p1) at (c.leftleg) {};
	  \node[cobowl] at (p1.belt) {};
	  \node[tube,anchor=top,yscale=2] (t1) at (p1.leftleg) {};
	  \node[copants,bot,anchor=leftleg] (c1) at (t1.bot) {};
	  \node[bowl] at (c1.belt) {};
	  \node[end_dot] at (c1.leftleg) {};
	  \begin{pgfonlayer}{nodelayer}
	    \node[label] (m) at (p.belt) {$\phi$};
	    \node[system] at ([yshift=0.6\ylength]c.rightleg) {$x'$};
	    \node[system] at ([yshift=-0.6\ylength]p.rightleg) {$x$};
	    \node[label] (f) at (p1.center) {$g$};
	    \node[label] (g) at (c1.center) {$f$};
	    \node[representable] (b) at (c1.belt) {$b$};
	    \node[representable] (b') at (p1.belt) {$b'$};
	  \end{pgfonlayer}
	  \begin{pgfonlayer}{strings}
	    \draw (m) to[out=-135,in=90] (p.leftleg) {};
	    \draw (m) to[out=-45,in=90] (p.rightleg) {};
	    \draw (m) to[out=135,in=-90] (c.leftleg) {};
	    \draw (m) to[out=45,in=-90] (c.rightleg) {};
	    \draw (f) to (b') {};
	    \draw (f) to[out=-25,in=90] (p1.rightleg) {};
	    \draw (f) to[out=-155,in=90] (p1.leftleg) to (c1.leftleg) to[out=-90,in=155] (g) {};
	    \draw (g) to[out=25,in=-90] (c1.rightleg) {};
	    \draw (g) to (b) {};
	  \end{pgfonlayer}
	\end{tikzpicture}
\end{xlrbox}

\begin{xlrbox}{supermap_comb3}
   \begin{tikzpicture}[baseline=(current bounding box.center),{every node/.style}={scale=0.8}]
	  \node[copants,bot,widebelt] (c) {};
	  \node[pants,bot,widebelt,anchor=belt] (p) at (c.belt) {};
	  \node[end] at (c.rightleg) {};
	  \node[cobowl] at (c.leftleg) {};
	  \node[bowl] at (p.leftleg) {};
	  \begin{pgfonlayer}{nodelayer}
	    \node[label] (m) at (p.belt) {$\phi$};
	    \node[system] at ([yshift=0.6\ylength]c.rightleg) {$x'$};
	    \node[system] at ([yshift=-0.6\ylength]p.rightleg) {$x$};
	    \node[label] (g) at ([xshift=-\tubewidth]p.center) {$f$};
	    \node[label] (f) at ([xshift=-\tubewidth]c.center) {$g$};
	    \node[representable] (b) at (c.leftleg) {$b'$};
	    \node[representable] (b') at (p.leftleg) {$b$};
	  \end{pgfonlayer}
	  \begin{pgfonlayer}{strings}
	    \draw (m) to[out=-135,in=45] (g) {};
	    \draw (m) to[out=-45,in=90] (p.rightleg) {};
	    \draw (m) to[out=135,in=-45] (f) {};
	    \draw (m) to[out=45,in=-90] (c.rightleg) {};
	    \draw (g) to[out=-135,in=90] (b') {};
	    \draw (g) to (f) {};
	    \draw (f) to[out=135,in=-90] (b) {};
	  \end{pgfonlayer}
	\end{tikzpicture}
\end{xlrbox}

\begin{example}
  Take morphisms $f:b \morph{} x\otimes a$ and $g:x \otimes a' \morph{} b'$, then there is a supermap of type \eqref{eq:basic_supermap} given applying the comb $(f,g)$ to the maps $\phi$.

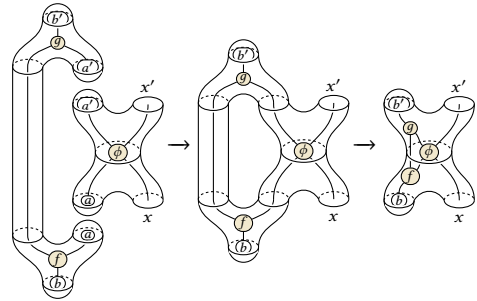
  \begin{equation*}
  	\xusebox{supermap_comb1} \morph{} \xusebox{supermap_comb2} \morph{} \xusebox{supermap_comb3}
  \end{equation*}
\end{example}

While this example is quite simple, the identification of single-party locally-applicable transformations with strong natural transformations allows us to generalise to the multi-party case, on which we will see there exist much more exotic examples.

\subsection{Unifying the Naturality and Strength Laws}
We now take a brief detour to explain how the laws \eqref{eq:nat_law} and \eqref{eq:strength_law} of locally-applicable transformations can be formally unified into the equivalent law \eqref{eq:commute_comb}.

Firstly note that the category $\Prof(\cat{C})$ of endoprofunctors and natural transformations can equivalently be described as the category of presheaves $[\opcat{\cat{C}}\times\cat{C},\set]$.

\begin{definition}
	A \textit{presheaf} on a category $\cat{C}$ is a functor $F:\opcat{\cat{C}}\morph{}\set$.
	A \textit{copresheaf} is a functor $G:\cat{C}\morph{}\set$.  
	We write $\psh{\cat{C}}$ for the category of presheaves on $\cat{C}$ and natural transformations between them.
	Similarly we write $\copsh{\cat{C}}$ for the category of copresheaves and natural transformations.
\end{definition}

One of the key results of \cite{pastro_street}, and a significant underpinning for the study of profunctor optics \cite{clarke_profunctor}, demonstrates that $\StProf(\cat{C})$ is also a category of copresheaves.

\begin{theorem}[\cite{pastro_street}]\label{thm:optics_presheaves}
  $\StProf(\cat{C}) \cong \copsh{\opt(\cat{C})}$.
\end{theorem}

\begin{definition}
  Let $(\cat{C},\otimes,i)$ be a symmetric monoidal category.
  The category $\opt(\cat{C})$ has objects given by pairs $\vb*{a}:=(a,a')$ of objects of $\cat{C}$.
  The homs are given by the following sets,
  \begin{equation*}
    \opt(\cat{C})(\vb*{a},\vb*{b}):= \int^x \cat{C}(b,x\otimes a)\times\cat{C}(x\otimes a',b').
  \end{equation*}
\end{definition}

We can think of the category $\opt(\cat{C})$ as having morphisms given by the single-holed combs built from morphisms of $\cat{C}$.
Explicitly, a morphism $\vb*{a}\morph{}\vb*{b}$ is given by an equivalence class of $(f,g)$ pairs of maps $f:b\morph{}x\otimes a$ and  $g:x\otimes a'\morph{}b'$ for some $x$, quotiented by the equivalence relation generated by sliding morphisms between the two halves of the comb.

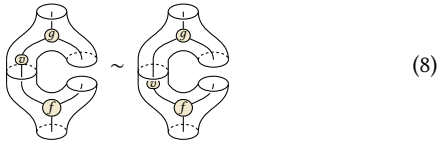
\begin{xlrbox}{optic1}
	\begin{tikzpicture}[baseline=(current bounding box.center),{every node/.style}={scale=0.8}]
		\node[pants,top] (p) {};
		\node[tube,yscale=0.5,anchor=top] (t) at (p.leftleg) {};
		\node[copants,bot,anchor=leftleg] (c) at (t.bot) {};
		\node[end] at (c.rightleg) {};
		\node[end_dot] at (t.center) {};
		\node[end_dot] at (p.rightleg) {};
		\begin{pgfonlayer}{nodelayer}
			\node[label] (g) at (p.center) {$g$};
			\node[label] (f) at (c.center) {$f$};
			\node[label] (v) at (p.leftleg) {$v$};
		\end{pgfonlayer}
		\begin{pgfonlayer}{strings}
			\draw (g) to[out=-155,in=90] (p.leftleg) {};
			\draw (g) to[out=-25,in=90] (p.rightleg) {};
			\draw (f) to[out=155,in=-90] (c.leftleg) {};
			\draw (f) to[out=25,in=-90] (c.rightleg) {};
			\draw (g) to (p.belt) {};
			\draw (f) to (c.belt) {};
			\draw (t.top) to (t.bot) {};
		\end{pgfonlayer}
	\end{tikzpicture}
\end{xlrbox}

\begin{xlrbox}{optic2}
	\begin{tikzpicture}[baseline=(current bounding box.center),{every node/.style}={scale=0.8}]
		\node[pants,top] (p) {};
		\node[tube,yscale=0.5,anchor=top] (t) at (p.leftleg) {};
		\node[copants,bot,anchor=leftleg] (c) at (t.bot) {};
		\node[end] at (c.rightleg) {};
		\node[end_dot] at (t.center) {};
		\node[end_dot] at (p.rightleg) {};
		\begin{pgfonlayer}{nodelayer}
			\node[label] (g) at (p.center) {$g$};
			\node[label] (f) at (c.center) {$f$};
			\node[label] (v) at (c.leftleg) {$v$};
		\end{pgfonlayer}
		\begin{pgfonlayer}{strings}
			\draw (g) to[out=-155,in=90] (p.leftleg) {};
			\draw (g) to[out=-25,in=90] (p.rightleg) {};
			\draw (f) to[out=155,in=-90] (c.leftleg) {};
			\draw (f) to[out=25,in=-90] (c.rightleg) {};
			\draw (g) to (p.belt) {};
			\draw (f) to (c.belt) {};
			\draw (t.top) to (t.bot) {};
		\end{pgfonlayer}
	\end{tikzpicture}
\end{xlrbox}

\begin{equation}\label{eq:sliding}
	\xusebox{optic1} \ \sim \ \xusebox{optic2}
\end{equation}

The category $\opt(\cat{C})$ has a monoidal product which acts on objects component-wise $(a,a')\otimes(b,b'):=(a\otimes a',b\otimes b')$ and on morphisms by placing combs next to each other in the intuitively obvious way \cite{riley_optics}.
The unit object is $(i,i)$.

Theorem \ref{thm:optics_presheaves} can be understood on an intuitive level as saying that to give a strong profunctor $P:\opcat{\cat{C}}\times\cat{C}\morph{}\set$ is to give a family of sets $P(a,b)$ that are not only functorial in $\opcat{\cat{C}}\times\cat{C}$, but are furthermore functorial over the category $\opt(\cat{C})$.
This means that the sets $P(a,b)$ have an action by combs from the category $\cat{C}$, combining both the functoriality and strength into a single notion.
This provides a formal connection between the laws \eqref{eq:nat_law} and \eqref{eq:strength_law} and the intuitive drawing \eqref{eq:basic_supermap}: a strong natural transformation is exactly one that commutes with combs.

Finally, it is worth noting that the Yoneda lemma ensures that we have embeddings of $\cat{C}$ into both its presheaves and copresheaves,
\begin{align*}
y^{-}:\cat{C} \morph{} \psh{\cat{C}} \\
y_{-}:\cat{C} \morph{} \opcat{\copsh{\cat{C}}}
\end{align*}
Explicitly these embeddings are given on objects by $y^a = \cat{C}(-,a)$ and $y_a=\cat{C}(a,-)$, and these presheaves are known as \textit{representable} and \textit{corepresentable} respectively.
This means that optics can be embedded into $\StProf(\cat{C})$ and thus the combs are present, and corepresentable, in the category in which all of the rest of our supermaps will live.
Some consequences of this will be investigated in Section \ref{sec:duals}.

\subsubsection{Combs vs.\ Optics}
We finish this section with a short discussion of a small but important technicality.
When we call the morphisms of $\opt(\cat{C})$ ``combs'' we must proceed with a little care for the category of optics is not the only way to define a category of combs over $\cat{C}$.
The equivalence relation generated by the sliding equations \eqref{eq:sliding} is not the only reasonable choice and is typically \textit{not} the one used in the study of quantum causal structures \cite{hefford_coend}.
This motivates the following definition where combs are equivalent up to operational, that is, extensional behaviour, as would be expected in the quantum literature.

\begin{definition}[\cite{hefford_coend}]
	Let $(\cat{C},\otimes,i)$ be a symmetric monoidal category.
	The category $\comb(\cat{C})$ has objects given by pairs $\vb*{a}=(a,a')$ of objects of $\cat{C}$.
	A morphism of $\comb(\cat{C})$ is an equivalence class of pairs $(f,g)$ of morphisms of $\cat{C}$ under the equivalence relation generated by saying that $(f,g)\sim(f',g')$ precisely when the combs are equal on all extended inputs.

	\begin{equation*}
		\tikzfigscale{1}{figs/comb_equiv1} \ = \ \tikzfigscale{1}{figs/comb_equiv2} \qquad \forall \phi
	\end{equation*}

	$\comb(\cat{C})$ has a monoidal operation which is completely analogous to $\opt(\cat{C})$.
\end{definition}

It is understood that in general $\opt(\cat{C})$ and $\comb(\cat{C})$ do not coincide, though there is always a full, identity on objects functor $\opt(\cat{C})\morph{}\comb(\cat{C})$ since the comb equivalence relation is more coarse than the optic equivalence relation.
For certain choices of $\cat{C}$ it has been shown that the two categories coincide, in particular when $\cat{C}$ is compact closed or $\cat{C}=\mathsf{Unitaries}$, the category of unitaries between Hilbert spaces \cite{hefford_coend}, but the general problem was left unresolved.

Here we establish that optics and combs do in fact coincide for the category $\cptp$ of quantum channels, i.e.\ completely positive trace preserving maps between Hilbert spaces.
This result along with the equivalent result for $\mathsf{Unitaries}$ are of great utility to us because they allow us to model quantum combs as their equivalent optics and thus in terms of strong profunctors.

\begin{theorem}\label{thm:cptp_equiv}
	There is a symmetric monoidal isomorphism of categories $\comb(\cptp)\cong\opt(\cptp)$.
\end{theorem}
\begin{proof}
	Given in Appendix \ref{sec:cptp_proof}.
\end{proof}

\section{Multi-Party Supermaps}

\subsection{Generalised Spaces of Maps}
The most basic form of supermap \eqref{eq:basic_supermap} defined in the previous section was allowed to take as its input any morphism on the system $(a,a')$ together with an arbitrary tensorial extension to the environment.
However it is known that not all supermaps can take arbitrary spaces of input channels.
This is especially true for multi-party supermaps which may only be well defined on, for example, separable inputs, or spaces with certain signalling constraints.

In general such constraints are required to avoid the creation of time-loops and thus to maintain causality.
The quantum switch, for instance, is not a well defined map on arbitrary bipartite inputs, generating time-loops on inputs like the swap.

This means we would like to be able to consider supermaps with (co)domains more general than those of the form $\suphom{a}{a'}$.
By taking seriously the observation that $\eta$ of \eqref{eq:basic_supermap} is a strong natural transformation between strong profunctors, we can use the structure of $\StProf(\cat{C})$ to generate many examples of interesting generalised spaces of maps.

\subsubsection{Tensorial Structure of Strong Profunctors}
The category of strong profunctors has two closed monoidal products that we can use to generate spaces of semi-localisable and separable channels.
To understand how these arise let us first ignore strength and consider the category $\Prof(\cat{C})$.
There are two closed monoidal operations on this category, the first is simply given by composition of endoprofunctors $Q\circ P$, with unit object given by the hom-profunctor $1_\cat{C}$.
The second comes from the underlying monoidal operation of the category $\cat{C}$.
There is an intimate connection between closed monoidal operations on a presheaf category $\psh{\cat{C}}$ and \textit{pro}monoidal operations on the underlying category $\cat{C}$.

\begin{definition}[\cite{day,day_thesis}]
	A category $\cat{C}$ is \textit{promonoidal} when it is equipped with a pair of profunctors,
	\begin{align*}
	  \otimes:\cat{C}\times\cat{C}\pmorph\cat{C}, \\
	  J: 1\pmorph\cat{C},
	\end{align*}
	and natural isomorphisms $\otimes(\otimes\times 1)\cong\otimes(1\times\otimes)$ and $\otimes(J\times 1)\cong 1\cong \otimes(1\times J)$ satisfying the triangle and pentagon coherence conditions.
\end{definition}
  
We can think of a promonoidal category as what we get when we interpret the data for a monoidal category in the bicategory $\Prof$ instead of $\Cat$.
This has the result of replacing the tensor product functor with a profunctor and unit object with a presheaf.
For applications to physics it may be useful to imagine this profunctorial tensor as an assignment of generalised ``virtual'' objects to tensors of $\cat{C}$.
These objects do not exist in $\cat{C}$ itself but are modelled on and inherit their properties from $\cat{C}$ \cite{hefford_spacetime}.

It is well understood by a classic theorem of Day that promonoidal structures and closed monoidal structures on presheaves coincide.

\begin{theorem}[\cite{day,day_thesis}]
	Promonoidal structures $(\otimes,J)$ on a category $\cat{C}$ are equivalent to closed monoidal structures on $\psh{\cat{C}}$.
\end{theorem}

Given a promonoidal category $(\cat{C},\otimes,J)$ the tensor product on $\psh{\cat{C}}$ is known as \textit{Day convolution} and is given by the following expression,
\begin{equation*}
  (F\star G)(-) = \int^{xy} \otimes(-,x,y) \times Fx \times Gy.
\end{equation*}
The unit object is the profunctor $J$ seen as a presheaf $J:\opcat{\cat{C}}\morph{}\set$.
The internal hom is given by the following expression,
\begin{equation*}
  [F,G](-) = \int_{xy} \set\big( \mathord{\otimes}(y,-,x), \set(Fx,Gy)\big).
\end{equation*}

Since $\Prof(\cat{C})=\copsh{\opcat{\cat{C}}\times\cat{C}}$ is a presheaf category, any promonoidal structures on $\opcat{\cat{C}}\times\cat{C}$ induce closed monoidal operations on $\Prof(\cat{C})$.
In particular, we can take Day convolution over the monoidal structure on $\opcat{\cat{C}}\times\cat{C}$, given by the following formula,
\begin{equation*}
  (P\otimes Q)(-,\bl) = \int^{abcd} \cat{C}(-,a\otimes b)\times P(a,c)\times Q(b,d) \times\cat{C}(c\otimes d,\bl).
\end{equation*}
The unit object is given by the Yoneda embedding $y_i\times y^i$ of the unit object $(i,i)$ from $\opcat{\cat{C}}\times\cat{C}$.

Going in the other direction, it is also possible to view the closed monoidal structure $(\circ,1_\cat{C})$ equivalently as a promonoidal structure on $\opcat{\cat{C}}\times\cat{C}$.
The profunctor giving the tensor product of this operation acts to produce a sequential pair of holes in the category $\cat{C}$, without environment wires.
\begin{equation*}
	(a,a'), (b,b') \mapsto \cat{C}(-,a)\times \cat{C}(a',b) \times \cat{C}(b'\bl)
\end{equation*}

Much of the discussion for $\Prof(\cat{C})$ proceeds analogously for $\StProf(\cat{C})$: it has two closed monoidal operations which are inherited from those of $\Prof(\cat{C})$.
The first is given by composition of the profunctors, together with the strength induced on this composition (i.e.\ by the universal property of the coend).
We will write $\seq$ for this tensor and the unit is still $1_\cat{C}$ equipped with its canonical strength given by taking tensor products $(f,g)\mapsto f\otimes g$.

The other tensor product arises by Day convolution over the tensor of $\opt(\cat{C})$.
We write $P\otimes_\cat{C}Q$ for this tensor and its unit is the Yoneda embedding of the unit from optics, $y_{(i,i)}\cong 1_\cat{C}$.
Thus it now coincides with the unit of the tensor $\seq$.

Let us unpack a little how the tensor $\otimes_\cat{C}$ behaves.
Recall that strong profunctors have a bimodular nature not only in the vertical but horizontal, tensorial direction.
This means they have a tensor product which is analogous to the tensor of bimodules, and to the coend, but in this horizontal direction.
Given two strong profunctors $P,Q:\cat{C}\pmorph\cat{C}$, their tensor $P\otimes_\cat{C}Q$ can also be described as a quotient of the Day convolution over $\opcat{\cat{C}}\times\cat{C}$, such that the right strength on $P$ and left strength on $Q$ are coequalised \cite{garner}.
If we consider the case of $1_\cat{C}\otimes_\cat{C} 1_\cat{C}$, this quotient has the effect of making the following diagrams equivalent.

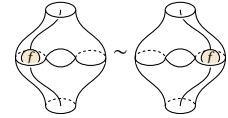
\begin{xlrbox}{hole1}
	\begin{tikzpicture}[baseline=(current bounding box.center),{every node/.style}={scale=0.8}]
		\node[pants,top,bot] (p) {};
		\node[copants,bot,anchor=leftleg] (c) at (p.leftleg) {};
		\begin{pgfonlayer}{nodelayer}
			\node[label] (f) at (p.leftleg) {$f$};
		\end{pgfonlayer}
		\begin{pgfonlayer}{strings}
			\draw (c.belt) to[out=90,in=-90] (f) to[out=90,in=-90] (p.belt) {};
		\end{pgfonlayer}
	\end{tikzpicture}
\end{xlrbox}

\begin{xlrbox}{hole2}
	\begin{tikzpicture}[baseline=(current bounding box.center),{every node/.style}={scale=0.8}]
		\node[pants,top,bot] (p) {};
		\node[copants,bot,anchor=leftleg] (c) at (p.leftleg) {};
		\begin{pgfonlayer}{nodelayer}
			\node[label] (f) at (p.rightleg) {$f$};
		\end{pgfonlayer}
		\begin{pgfonlayer}{strings}
			\draw (c.belt) to[out=90,in=-90] (f) to[out=90,in=-90] (p.belt) {};
		\end{pgfonlayer}
	\end{tikzpicture}
\end{xlrbox}

\begin{equation*}
	\xusebox{hole1} \sim \xusebox{hole2}
\end{equation*}

This equivalence is analogous to the sliding equations for the coend \eqref{eq:sliding}, and thus also to the tensor product of bimodules.
It is also intuitively clear from the above diagram why $1_\cat{C}$ is the unit of $\otimes_\cat{C}$. 

\begin{xlrbox}{optic2hole}
	\begin{tikzpicture}[baseline=(current bounding box.center),{every node/.style}={scale=0.8}]
		\node[pants,top] (p) {};
		\node[tube,yscale=0.5,anchor=top] (t) at (p.leftleg) {};
		\node[copants,bot,anchor=leftleg] (c) at (t.bot) {};
		\node[end] at (c.rightleg) {};
		\node[end_dot] at (t.center) {};
		\node[end_dot] at (p.rightleg) {};
		\node[pants,anchor=belt] (p2) at (c.belt) {};
		\node[tube,yscale=0.5,anchor=top] (t2) at (p2.leftleg) {};
		\node[copants,bot,anchor=leftleg] (c2) at (t2.bot) {};
		\node[end] at (c2.rightleg) {};
		\node[end_dot] at (t2.center) {};
		\node[end_dot] at (p2.rightleg) {};
		\begin{pgfonlayer}{nodelayer}
			\node[label] (g) at (p.center) {$k$};
			\node[label] (f) at (c.center) {$h$};
			\node[label] (g2) at (p2.center) {$g$};
			\node[label] (f2) at (c2.center) {$f$};
		\end{pgfonlayer}
		\begin{pgfonlayer}{strings}
			\draw (g) to[out=-155,in=90] (p.leftleg) {};
			\draw (g) to[out=-25,in=90] (p.rightleg) {};
			\draw (f) to[out=155,in=-90] (c.leftleg) {};
			\draw (f) to[out=25,in=-90] (c.rightleg) {};
			\draw (g) to (p.belt) {};
			\draw (f) to (c.belt) {};
			\draw (t.top) to (t.bot) {};
			\draw (g2) to[out=-155,in=90] (p2.leftleg) {};
			\draw (g2) to[out=-25,in=90] (p2.rightleg) {};
			\draw (f2) to[out=155,in=-90] (c2.leftleg) {};
			\draw (f2) to[out=25,in=-90] (c2.rightleg) {};
			\draw (g2) to (p2.belt) {};
			\draw (f2) to (c2.belt) {};
			\draw (t2.top) to (t2.bot) {};
		\end{pgfonlayer}
	\end{tikzpicture}
\end{xlrbox}

It is worth understanding the promonoidal structure on $\opt(\cat{C})$ induced by the closed monoidal structure $(\seq,1_\cat{C})$ on $\StProf(\cat{C})$.
On objects it is given by the following expression (rotated to save space),
\begin{equation*}
  \seq(\vb*{a},\vb*{b},\vb*{c}) = \int^{xy} \cat{C}(c,x\otimes a)\times\cat{C}(x\otimes a',y\otimes b)\times\cat{C}(y\otimes b',c').
\end{equation*}
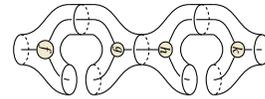
\begin{equation*}
	\rotatebox{-90}{\xusebox{optic2hole}}
\end{equation*}
We can think of this as the set of two-holed combs with inputs $\vb*{a}$ and $\vb*{b}$ and output $\vb*{c}$.
In fact, the promonoidal structure $\circ$ generates all the higher arity combs by composing it with itself.

\begin{definition}
	A category $\cat{C}$ is \textit{duoidal} when it is equipped with two monoidal structures $(\seq,i)$ and $(\otimes,j)$ and natural transformations of the following form,
	\begin{align*}
		(a\seq b) \otimes (c\seq d) \morph{} (a\otimes c)\seq (b\otimes d) \\
		j \morph{} j\seq j, \qquad i \otimes i \morph{} i, \qquad j \morph{} i
	\end{align*}
	These must satisfy a series of coherence conditions which can be found in e.g.\ \cite{garner}.
	A duoidal category is \textit{normal} when $i\cong j$, and \textit{closed} when both tensor products are closed.
\end{definition}
  
The two tensors $(\circ,1_\cat{C})$ and $(\otimes,y_i\times y^i)$ interact so as to make $\Prof(\cat{C})$ into a closed duoidal category \cite{garner}.
It can furthermore be shown that the duoidal structure of $\Prof(\cat{C})$ is inherited by $\StProf(\cat{C})$ so that it is a closed normal duoidal category with respect to its two monoidal structures $(\seq,1_\cat{C})$ and $(\otimes_\cat{C},1_\cat{C})$. 
In fact, $\StProf(\cat{C})$ is precisely the \textit{normalisation} of the duoidal category $\Prof(\cat{C})$ \cite{garner}.

We can equivalently view the closed duoidal structure of $\Prof(\cat{C})$ and $\StProf(\cat{C})$ as \textit{pro}duoidal structures on $\opcat{C}\times\cat{C}$ and $\opt(\cat{C})$ respectively \cite{earnshaw}.
These each have two promonoidal structures that interact via identical laws as for a duoidal category, just interpreted in $\Prof$.
This because there is an equivalence between produoidal structures on a category and closed duoidal structures on its category of presheaves \cite{booker}.

Finally we note that the category $\comb(\cat{C})$ possesses an analogous promonoidal operation which returns the two-holed combs, making this category produoidal.
In the case of $\cptp$ we extend Theorem \ref{thm:cptp_equiv} to show that the isomorphism includes the produoidal structure and thus for quantum theory we can treat $n$-combs as $n$-optics.

\begin{theorem}[extending Theorem \ref{thm:cptp_equiv}]\label{thm:cptp_equiv2}
	There is an isomorphism of produoidal categories $\comb(\cptp)\cong\opt(\cptp)$.
\end{theorem}
\begin{proof}
	Given in Appendix \ref{sec:cptp_proof2}.
\end{proof}

\subsubsection{Generalised Spaces of Maps from the Tensors of Strong Profunctors}
We can use the tensors $\seq$ and $\otimes_\cat{C}$ to generate further examples of strong profunctors which will act as the (co)domains of multi-party supermaps

\begin{xlrbox}{semi_localisable}
	\begin{tikzpicture}[baseline=(current bounding box.center),{every node/.style}={scale=0.8}]
		\node[copants,bot,widebelt] (c) {};
    \node[pants,bot,widebelt,anchor=belt] (p) at (c.belt) {};
    \node[copants,bot,widebelt,anchor=rightleg] (c2) at (p.leftleg) {};
    \node[pants,bot,widebelt,anchor=belt] (p2) at (c2.belt) {};
    \node[cobowl] at (c.rightleg) {};
    \node[bowl] at (p.rightleg) {};
    \node[cobowl] at (c2.leftleg) {};
    \node[bowl] at (p2.leftleg) {};
    \node[end] at (c.leftleg) {};
		\begin{pgfonlayer}{nodelayer}
      \node[label] (f) at (p2.belt) {$f$};
      \node[label] (g) at (p.belt) {$g$};
      \node[representable] at (c.rightleg) {$b'$};
      \node[representable] at (p.rightleg) {$b$};
      \node[representable] at (c2.leftleg) {$a'$};
      \node[representable] at (p2.leftleg) {$a$};
		\end{pgfonlayer}
		\begin{pgfonlayer}{strings}
      \draw (f) to[out=125,in=-90] (c2.leftleg) {};
      \draw (f) to[out=-125,in=90] (p2.leftleg) {};
      \draw (f) to[out=-55,in=90] (p2.rightleg) {};
      \draw (f) to[out=55,in=-90] (c2.rightleg) to[out=90,in=-125] (g) {};
      \draw (g) to[out=125,in=-90] (c.leftleg) {};
      \draw (g) to[out=-55,in=90] (p.rightleg) {};
      \draw (g) to[out=55,in=-90] (c.rightleg) {};
		\end{pgfonlayer}
	\end{tikzpicture}
\end{xlrbox}

\begin{example}[Semi-localisable]
  Taking the sequential tensor of the basic types $\suphomrev{b}{b'}\seq \suphom{a}{a'}$ produces a space of semi-localisable channels with party $\vb*{a}$ causally preceding $\vb*{b}$ (left of \eqref{eq:gen_spaces}).
  As a result $\vb*{a}$ is able to signal to $\vb*{b}$ but not the other way around.
  Note that it is \textit{not} the case in a general category $\cat{C}$ that the one-way signalling channels are semi-localisable and factorise in the causally faithful way written above.
  There is an active area of research studying the relationship between signalling conditions and causally faithful factorisations \cite{beckman_localisable,schumacher_locality,lorenz_unitary}.
  One reason we deal here with semi-localisable channels instead of one-way signalling channels is to be as theory independent as possible; we do not assume that $\cat{C}$ has the necessary features (e.g.\ an environment structure \cite{coecke_environment,coecke2016terminality}) to even define one-way signalling.
\end{example}

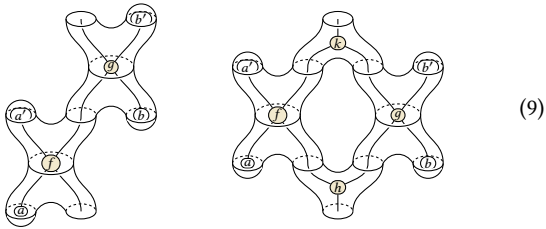
\begin{xlrbox}{separable}
	\begin{tikzpicture}[baseline=(current bounding box.center),{every node/.style}={scale=0.8}]
    \node[pants,top,bot] (p) {};
		\node[copants,bot,widebelt,anchor=leftleg] (c) at (p.rightleg) {};
    \node[pants,bot,widebelt,anchor=belt] (p2) at (c.belt) {};
    \node[copants,bot,widebelt,anchor=rightleg] (c2) at (p.leftleg) {};
    \node[pants,bot,widebelt,anchor=belt] (p3) at (c2.belt) {};
    \node[copants,bot,anchor=rightleg] (c3) at (p2.leftleg) {};
    \node[cobowl] at (c.rightleg) {};
    \node[bowl] at (p2.rightleg) {};
    \node[cobowl] at (c2.leftleg) {};
    \node[bowl] at (p3.leftleg) {};
		\begin{pgfonlayer}{nodelayer}
      \node[label] (f) at (p3.belt) {$f$};
      \node[label] (g) at (p2.belt) {$g$};
      \node[representable] at (c.rightleg) {$b'$};
      \node[representable] at (p2.rightleg) {$b$};
      \node[representable] at (c2.leftleg) {$a'$};
      \node[representable] at (p3.leftleg) {$a$};
      \node[label] (h) at (c3.center) {$h$};
      \node[label] (k) at (p.center) {$k$};
		\end{pgfonlayer}
		\begin{pgfonlayer}{strings}
      \draw (f) to[out=125,in=-90] (c2.leftleg) {};
      \draw (f) to[out=-125,in=90] (p3.leftleg) {};
      \draw (f) to[out=-55,in=90] (p3.rightleg) to[out=-90,in=155] (h) {};
      \draw (f) to[out=55,in=-90] (c2.rightleg) to[out=90,in=-155] (k) {};
      \draw (g) to[out=125,in=-90] (c.leftleg) to[out=90,in=-25] (k) {};
      \draw (g) to[out=-55,in=90] (p2.rightleg) {};
      \draw (g) to[out=55,in=-90] (c.rightleg) {};
      \draw (g) to[out=-125,in=90] (p2.leftleg) to[out=-90,in=25] (h) {};
      \draw (h) to (c3.belt) {};
      \draw (k) to (p.belt) {};
		\end{pgfonlayer}
	\end{tikzpicture}
\end{xlrbox}

\begin{xlrbox}{separable_quotient1}
	\begin{tikzpicture}[baseline=(current bounding box.center),{every node/.style}={scale=0.8}]
    \node[pants,top,bot] (p) {};
		\node[copants,bot,widebelt,anchor=leftleg] (c) at (p.rightleg) {};
    \node[pants,bot,widebelt,anchor=belt] (p2) at (c.belt) {};
    \node[copants,bot,widerbelt,anchor=rightleg] (c2) at (p.leftleg) {};
    \node[pants,bot,widerbelt,anchor=belt] (p3) at (c2.belt) {};
    \node[copants,bot,anchor=rightleg] (c3) at (p2.leftleg) {};
    \node[cobowl] at (c.rightleg) {};
    \node[bowl] at (p2.rightleg) {};
    \node[cobowl] at (c2.leftleg) {};
    \node[bowl] at (p3.leftleg) {};
		\begin{pgfonlayer}{nodelayer}
      \node[label] (f) at ([xshift=-.75\tubewidth]p3.belt) {$f$};
      \node[label] (phi) at ([xshift=.75\tubewidth]p3.belt) {$\phi$};
      \node[label] (g) at (p2.belt) {$g$};
      \node[representable] at (c.rightleg) {$b'$};
      \node[representable] at (p2.rightleg) {$b$};
      \node[representable] at (c2.leftleg) {$a'$};
      \node[representable] at (p3.leftleg) {$a$};
      \node[label] (h) at (c3.center) {$h$};
      \node[label] (k) at (p.center) {$k$};
		\end{pgfonlayer}
		\begin{pgfonlayer}{strings}
      \draw (f) to[out=125,in=-90] (c2.leftleg) {};
      \draw (f) to[out=-125,in=90] (p3.leftleg) {};
      \draw (f) to[out=-65,in=90] ([xshift=-.3\tubewidth]p3.rightleg) to[out=-90,in=155] (h) {};
      \draw (f) to[out=65,in=-90] ([xshift=-.3\tubewidth]c2.rightleg) to[out=90,in=-155] (k) {};
      \draw (g) to[out=125,in=-90] (c.leftleg) to[out=90,in=-25] (k) {};
      \draw (g) to[out=-55,in=90] (p2.rightleg) {};
      \draw (g) to[out=55,in=-90] (c.rightleg) {};
      \draw (g) to[out=-125,in=90] (p2.leftleg) to[out=-90,in=25] (h) {};
      \draw (h) to (c3.belt) {};
      \draw (k) to (p.belt) {};
      \draw (phi) to[out=-90,in=90] ([xshift=.3\tubewidth]p3.rightleg) to[out=-90,in=135] (h) {};
      \draw (phi) to[out=90,in=-90] ([xshift=.3\tubewidth]c2.rightleg) to[out=90,in=-135] (k) {};
		\end{pgfonlayer}
	\end{tikzpicture}
\end{xlrbox}

\begin{xlrbox}{separable_quotient2}
	\begin{tikzpicture}[baseline=(current bounding box.center),{every node/.style}={scale=0.8}]
    \node[pants,top,bot] (p) {};
		\node[copants,bot,widerbelt,anchor=leftleg] (c) at (p.rightleg) {};
    \node[pants,bot,widerbelt,anchor=belt] (p2) at (c.belt) {};
    \node[copants,bot,widebelt,anchor=rightleg] (c2) at (p.leftleg) {};
    \node[pants,bot,widebelt,anchor=belt] (p3) at (c2.belt) {};
    \node[copants,bot,anchor=rightleg] (c3) at (p2.leftleg) {};
    \node[cobowl] at (c.rightleg) {};
    \node[bowl] at (p2.rightleg) {};
    \node[cobowl] at (c2.leftleg) {};
    \node[bowl] at (p3.leftleg) {};
		\begin{pgfonlayer}{nodelayer}
      \node[label] (f) at (p3.belt) {$f$};
      \node[label] (phi) at ([xshift=-.75\tubewidth]p2.belt) {$\phi$};
      \node[label] (g) at ([xshift=.75\tubewidth]p2.belt) {$g$};
      \node[representable] at (c.rightleg) {$b'$};
      \node[representable] at (p2.rightleg) {$b$};
      \node[representable] at (c2.leftleg) {$a'$};
      \node[representable] at (p3.leftleg) {$a$};
      \node[label] (h) at (c3.center) {$h$};
      \node[label] (k) at (p.center) {$k$};
		\end{pgfonlayer}
		\begin{pgfonlayer}{strings}
      \draw (f) to[out=125,in=-90] (c2.leftleg) {};
      \draw (f) to[out=-125,in=90] (p3.leftleg) {};
      \draw (f) to[out=-55,in=90] (p3.rightleg) to[out=-90,in=155] (h) {};
      \draw (f) to[out=55,in=-90] (c2.rightleg) to[out=90,in=-155] (k) {};
      \draw (g) to[out=115,in=-90] ([xshift=.3\tubewidth]c.leftleg) to[out=90,in=-25] (k) {};
      \draw (g) to[out=-55,in=90] (p2.rightleg) {};
      \draw (g) to[out=55,in=-90] (c.rightleg) {};
      \draw (g) to[out=-115,in=90] ([xshift=.3\tubewidth]p2.leftleg) to[out=-90,in=25] (h) {};
      \draw (h) to (c3.belt) {};
      \draw (k) to (p.belt) {};
      \draw (phi) to[out=-90,in=90] ([xshift=-.3\tubewidth]p2.leftleg) to[out=-90,in=45] (h) {};
      \draw (phi) to[out=90,in=-90] ([xshift=-.3\tubewidth]c.leftleg) to[out=90,in=-45] (k) {};
		\end{pgfonlayer}
	\end{tikzpicture}
\end{xlrbox}

\begin{equation}\label{eq:gen_spaces}
    \xusebox{semi_localisable} \hspace{1cm} \xusebox{separable}
\end{equation}

\begin{example}[Separable]
  Taking the horizontal tensor of the basic types $\suphom{a}{a'}\otimes_\cat{C}\suphomrev{b}{b'}$ produces a space of channels where the parties $\vb*{a}$ and $\vb*{b}$ are each allowed to signal with the environment but not each other.
  This space consists of the family of sets on the right of \eqref{eq:gen_spaces}, quotiented by the following relation allowing maps that are separable from the parties $\vb*{a}$ and $\vb*{b}$ to pass between the two halves.
  \begin{equation*}
    \xusebox{separable_quotient1} \sim \xusebox{separable_quotient2}
  \end{equation*}
  Since only the separable $\phi$ can pass between the halves, the above set does not generally contain all the morphisms in $\suphom{a\otimes b}{a'\otimes b'}$.
\end{example}

The normal duoidal structure of $\StProf(\cat{C})$ gives distributors $(P\seq Q) \otimes_\cat{C} (R\seq S) \morph{} (P\otimes_\cat{C} R)\seq (Q\otimes_\cat{C} S)$ between separable and semi-localisable spaces.
These also induce linear distributors \cite{cockett}
\begin{align*}
  (P\seq Q)\otimes_\cat{C}R \morph{} P\seq(Q\otimes_\cat{C} R),\\
  (P\seq Q)\otimes_\cat{C}R \morph{} (P\otimes_\cat{C} R) \seq Q,
\end{align*}
and in turn embeddings $P\otimes_\cat{C}R \morph{} P\seq R$ and $P\otimes_\cat{C} R\morph{} R\seq P$.
Much like in the $\mathsf{Caus}$-construction these distributors give us embeddings of separable channels into semi-localisable channels, and between more complex spaces built from these.

\subsection{Indefinite Causal Order from the Separable Tensor Product}

The most important examples of supermaps and those which really motivate their study in the quantum foundations literature, are the supermaps which exhibit indefinite causal order such as the switch \cite{chiribella_switch}, Lugano process \cite{baumeler_logically, baumeler_incompatibility}, and Grenoble process \cite{wechs}. 
These multi-partite supermaps are often drawn intuitively as boxes with several parallel holes, each of which can take an input with an extension to the environment.

The important thing to note is that one cannot act the supermap on arbitrary bipartite processes between the holes, else it would be possible to generate time-loops.

These multi-partite supermaps can be modelled using a multi-input generalisation of the diagrammatic laws \eqref{eq:nat_law} and \eqref{eq:strength_law}.

\begin{definition}[\cite{wilson_locality}]
	A multi-partite locally-applicable transformation $\eta:\vb*{a}_1,\dots,\vb*{a}_n\morph{}\vb*{b}$ consists of a family of functions of the form 
	\begin{equation*}
		\tikzfigscale{1}{figs/multi_supermap_family}
	\end{equation*}
	Such that the following law holds ensuring that the supermap commutes with pre- and post composition on each environment,
	\begin{equation}\label{eq:multi_nat_law}
		\tikzfigscale{1}{figs/multi_supermap_nat1} \ = \ \tikzfigscale{1}{figs/multi_supermap_nat2}
	\end{equation}
	and the following law holds \textbf{for each} of the $n$ holes, ensuring that further tensorial extension of the environment can be pulled into any one of the holes
	\begin{equation}\label{eq:multi_strength_law}
		\tikzfigscale{1}{figs/multi_supermap_strength1} = \tikzfigscale{1}{figs/multi_supermap_strength2},
	\end{equation}
	where this equation holds up to the obvious symmetry required.
\end{definition}

One issue with the previous definition is that it is not entirely clear what the domain of these supermaps should be.
In \cite{wilson_polycategories} the domain $\vb*{a}_1,\dots,\vb*{a}_n$ is treated as a list and the multi-partite locally applicable transformations are shown to form a multicategory $\mathbf{lot}(\cat{C})$ with composition given by nesting of supermaps.

By making the identification between the laws \eqref{eq:nat_law} and \eqref{eq:strength_law} and strong natural transformations, we are able to use the structure of the category $\StProf(\cat{C})$ to give an alternative definition of multi-partite locally-applicable transformations.

\begin{theorem}\label{thm:separable_supermaps}
	On any symmetric monoidal category $\cat{C}$, the strong natural transformations of type 
	\begin{equation*}
		\eta: \modtensor^{\hspace{-0.7em}n}_{\hspace{-0.7em}i=1} \suphom{a_i}{a_i'} \rightarrow \suphom{b}{b'}
	\end{equation*}
	are the multi-partite locally-applicable transformations of type $\eta:\vb*{a}_1,\dots,\vb*{a}_n\morph{}\vb*{b}$.
\end{theorem}
\begin{proof}
	Given in Appendix \ref{sec:separable_supermaps_proof}.
\end{proof}

Theorem \ref{thm:separable_supermaps} demonstrates that it is possible to reduce the domain $\vb*{a}_1,\dots,\vb*{a}_n$ of a multi-partite locally-applicable transformation to a single object of $\StProf(\cat{C})$ using the tensor product $\otimes_\cat{C}$.
One may wonder whether this makes the multicategory $\mathbf{lot}(\cat{C})$ representable \cite{hermida_repmulticats}.
The answer is no, and for one simple reason, the objects of the form $\otimes_\cat{C} \suphom{a_i}{a_i'}$ were not included in the original definition of $\mathbf{lot}(\cat{C})$ because the connection with strong profunctors was not known!
By adding these in and considering the full subcategory of $\StProf(\cat{C})$ spanned by the objects of the aforementioned form, one yields a symmetric monoidal category $\mathsf{lot}(\cat{C})$ which is a good replacement for the multicategory $\mathbf{lot}(\cat{C})$.
Indeed, there is a full, faithful, identity on objects embedding of $\mathbf{lot}(\cat{C})$ into the representable multicategory corresponding to $\mathsf{lot}(\cat{C})$.

The identification of multi-partite locally-applicable transformations and this class of strong natural transformations has two key consequences.
First, it can be used to show that multi-input quantum superchannels are exactly the morphisms of strong profunctors on the separable tensor product, opening the door to fully categorical analysis of indefinite causal (or more generally compositional) order.

\begin{theorem}
The quantum supermaps on the non-signalling channels are the morphisms of strong profunctors of type \[S: \bigotimes^{i}_{\cptp} \cptp(a_i \otimes -, a_i' \otimes \bl) \rightarrow \cptp(b \otimes -, b' \otimes \bl).\]
\end{theorem}

\begin{proof}
	From Theorem \ref{thm:separable_supermaps} the morphisms of the above type can be identified with a unique multi-input locally-applicable transformations.
	Following \cite{wilson_locality} each multi-input locally-applicable transformation can be identified with a unique multi-input quantum supermap. 
\end{proof}

Second, this theorem can be used in practice for the purpose of constructing examples on the quite abstract separable tensor product space.
Let us show that for any symmetric monoidal category enriched in convex spaces \cite{stone1949postulates,neumann1970quasivariety,fritz2015convex}, there is a morphism of strong profunctors of this type which is a probabilistic switch. 
\begin{example}
    Let $\cat{C}$ be symmetric monoidally enriched in convex spaces, then there exists a morphism of strong profunctors $\textrm{Switch}: \suphom{a}{a} \otimes_{\cat{C}} \suphom{a}{a} \rightarrow \suphom{a}{a}$. This morphism acts as follows on class representatives. \[\textrm{Switch}(\phi_1,\phi_2) = 
	c_p\left( \tikzfigscale{1}{figs/c_switch_1} , \tikzfigscale{1}{figs/c_switch_2} \right).
\] 
Note that when hom-objects from enrichment here are cancellative, convex combinations may be written more familiarly using convex sums \cite{stone1949postulates} as $c_{p}(f,g) = (1-p)f + pg$. 
\end{example}

\subsection{Definite Causal Order from the Semi-Localisable Tensor Product}

In quantum information theory there are two equivalent formulations of the study of definite causal order \cite{chiribella_networks}. The first, the constructive approach, consists in defining networks of circuit diagrams into which holes have been punctured.
This is the viewpoint on definite causal structures which has been well modelled in the categorical literature using equivalences with respect to operational behaviour \cite{hefford_coend, coecke_resources} and coequalisation \cite{roman_coend, hefford_coend}. 

The second viewpoint on definite causal structure however, has yet to be studied categorically because it relies on the notion of a black-box environment.
In this setting, referred to as the \textit{axiomatic} approach (or the approach of \textit{admissability}) \cite{chiribella_networks} definite causal structures with $n+1$ holes are defined to be the boxes with holes into which $n$-hole definite causal structure can be inserted.
There is, in fact, an even weaker notion which can be written down, the notion of an $n+1$-hole definite causal structure as one into which a concrete circuit diagram with $n$-holes can be inserted.

In quantum theory all of these notions are equivalent, however, there are examples of symmetric monoidal categories in which not even all one-input supermaps decompose as combs \cite{kissinger_caus}.
In this section we find that the abstract notion of supermaps (without a concrete decomposition into nodes) on definitely causally ordered channels, can be captured using morphisms out of the semi-localisable tensor product $\seq$. 
In the case of quantum theory, the category $\cptp$, these morphisms exactly recover the quantum definite causal orders (whether they be equivalently defined as supermaps on no-signalling channels or as concrete combs). 

First, let us see what morphisms out of $\seq$ concretely encode.
Perhaps unsurprisingly, they admit a diagrammatic representation in which they look very much like black-box operations which can be applied to parts of circuit diagrams. 
\begin{lemma}\label{lem:seq}
	On any symmetric monoidal category $\cat{C}$ the morphisms of type \[ \eta : \suphom{b}{b'} \seq \suphom{a}{a'} \rightarrow\suphom{c}{c'} \] in $\StProf(\cat{C})$ are the families of functions of type \[\eta_{x,x',z}: \cat{C}(a \otimes x, a' \otimes z)\times \cat{C}(b \otimes z , b' \otimes x') \rightarrow \cat{C}(c \otimes x, c' \otimes x')\] satisfying the following two equations.
	\begin{equation}\label{eq:vert_naturality}
	\tikzfigscale{1}{figs/multi_vtensor_diagram_1} \ = \ \tikzfigscale{1}{figs/multi_vtensor_diagram_2}
	\end{equation}
	\begin{equation}\label{eq:vert_dinaturality}
	\tikzfigscale{1}{figs/multi_vtensor_diagram_3} \ = \ \tikzfigscale{1}{figs/multi_vtensor_diagram_4}
	\end{equation}
\end{lemma}

\begin{proof}
	Given in Appendix \ref{sec:seq_proof}.
\end{proof}

Lemma \ref{lem:seq} makes it clear that supermaps out of the tensor product $\seq$ possess a notion of definite causal structure.
If it were the case that $\eta$ could behave like a trace, taking the output of $\phi_2$ back into the input of $\phi_1$, then it would be possible to form a time-loop by choosing both $\phi_1$ and $\phi_2$ to be the swap.
Since the underlying category is not assumed to be compact closed or traced, there would be no good way to interpret such an outcome, meaning at least intuitively that supermaps of this type must have a definite order: information cannot flow back from $\phi_2$ to $\phi_1$, and there is a definite ordering to the holes in the supermap.

Lemma \ref{lem:seq} also simplifies the construction of explicit examples of supermaps on the tensor product $\seq$, by unpacking their behaviour into a pair of laws that generalise the original laws given for locally applicable transformations \cite{wilson_locality}.
We see here the utility of the categorical framework: these laws, while perhaps obvious in retrospect, were not known and appear for free out of the framework.
The good news is that they are also the right laws in the sense that they capture the objects of interest in the study of protocols with definite causal structure in quantum information theory. 

\begin{theorem}\label{thm:seq_supermaps}
	The quantum supermaps on the $n$-combs are the morphisms in $\StProf(\cptp)$ of the following type,
	\begin{equation}\label{eq:supermap_seq_type}
		 \bigseq_{i=1}^n  \cptp(a_i \otimes - ,a_i' \otimes \bl) \rightarrow \cptp(c \otimes - ,c' \otimes \bl).
	\end{equation}
\end{theorem}
\begin{proof}
	Given in Appendix \ref{sec:seq_supermaps_proof}.
\end{proof}

This establishes that just at the study of $\otimes_{\cat{C}}$ categorifies the study of indefinite causal structure, the study of $\seq$ categorifies the study of definite causal structure. 

\section{Decomposition and Duality}\label{sec:duals}

Analogously to the core realisation theorem of standard quantum information theory, which says that every quantum channel can be decomposed as a partial trace over an isometry, maybe the fundamental theorem of quantum supermaps is the discovery that the primitive 1-input supermaps can be concretely realised in terms of a simple network with a hole \cite{chiribella_supermaps}.

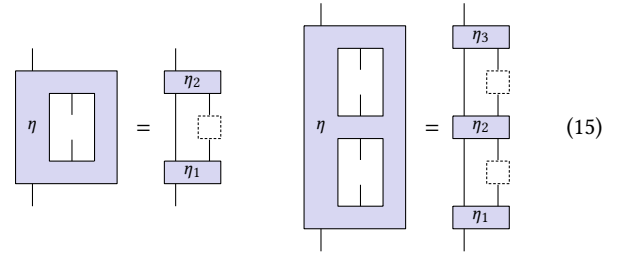
\begin{equation}
	\tikzfigscale{1}{figs/decomp_single_1} \ = \ \tikzfigscale{1}{figs/decomp_single_2} \hspace{1cm} \tikzfigscale{1}{figs/vertical_physics_1} \ = \ \tikzfigscale{1}{figs/vertical_physics_3}
\end{equation}

In fact, this extends to a theorem on all black-box supermaps with a definite causal order.
Any quantum supermap which can be applied to a network, can itself be decomposed as a concrete network of channels and holes.

In the specific setting of quantum information theory, objects on the left hand side are defined as linear maps satisfying nice \textit{preservation} properties (analogous to complete positivity) and objects on the right are interpreted as again linear maps, this time satisfying a nice \textit{decompositonal} property (analogous to the dilation theorem of Stinespring).

\subsection{Supermap Decomposition Theorems as Representability}

In the category $\StProf(\cat{C})$ we now have sufficient tools at our disposal to give a categorical interpretation of decomposition theorems for supermaps.
First, we have a notion of black-box supermap with definite causal order as a morphism with domain $ \seq_i \suphom{a_i}{a_i'}$.
Second, we are able to reference the most fine-grained notion of a network decomposition by the Yoneda embedding of optics in $\StProf(\cat{C})$. 
The Yoneda lemma ensures that $\StProf(\cat{C})(y_{b,b'}, y_{a,a'})$ is naturally isomorphic to $\opt(\cat{C})((a,a') , (b,b'))$ and so the space of one-holed optics is easily accessible in $\StProf(\cat{C})$.
Similarly, the Yoneda lemma implies that the $n$-holed optics are naturally isomorphic to $\StProf(\cat{C})(y_{b,b'}, \seq_i y_{a_i,a_i'})$.

\begin{definition}[Supermap Decomposition Theorem]
A symmetric monoidal category $\cat{C}$ has a $1$-arity supermap decomposition theorem if there is a natural isomorphism
	\begin{align*}
		& \StProf(\cat{C})\big(\suphom{a}{a'},\suphom{b}{b'}\big) \\
		& \cong  \StProf(\cat{C})\big(y_{b,b'},y_{a,a'} \big).
	\end{align*}
More generally $\cat{C}$ has an $n$-arity supermap decomposition theorem if there is a natural isomorphism
	\begin{align*}
		& \StProf(\cat{C})\big( \seq_i \suphom{a_i}{a_i'},\suphom{b}{b'}\big) \\
		& \cong \StProf(\cat{C})(y_{b,b'}, \seq_i y_{a_i,a_i'}).
	\end{align*}
\end{definition}

So, we can phrase the notion of a decomposition theorem for black box supermaps as representability: it is the requirement that one could equivalently have considered those supermaps as acting on representable presheaves, and so by the Yoneda lemma would be equivalent to the concrete morphisms of the underlying category $\opt(\cat{C})$.

Of course, it is important to check that this categorified decomposition property holds for the most important working example, and indeed it does. 
\begin{theorem}\label{thm:cptp_decomp}
The category $\cptp$ has an $n$-arity supermap decomposition theorem for every $n$. 
\end{theorem}
\begin{proof}
	Given in Appendix \ref{sec:cptp_decomp_proof}.
\end{proof}

\begin{remark}
	The previous proof almost holds for the category $\mathsf{Stoch}$ of stochastic matrices.
	It is clear in this setting that morphisms of strong profunctors out of the $\seq$ tensor can also be represented by concrete classical superchannels \cite{wilson_locality}, which furthermore admit a concrete decomposition into combs \cite{chiribella_networks}, yet it is unclear at this stage whether $\opt(\mathsf{Stoch}) \cong \comb(\mathsf{Stoch})$.
	Consequently it is unclear whether the category of deterministic classical channels has a fully categorical supermap decomposition theorem, though deciding on the equivalency of optics and combs would resolve the matter.
\end{remark}

\subsection{Duals in Strong Endoprofunctors}
One of the key applications of constructing connectives such as $(- \otimes_{\cat{C}} - )$ and $( - \seq - )$ is to find safe and increasingly complex ways of reasoning about the quite abstract conceptual objects which are supermaps.
In short, one key reason to develop a categorical framework for supermaps is to find nice logics for abstract supermaps.
With this in mind, a reader familiar with the $\caus$-construction may recall that it sends any compact closed (and pre-causal) category to a $*$-autonomous category. 

The $*$ of a $*$-autonomous category can always be phrased in terms of its closed monoidal structure, and $\StProf(\cat{C})$ is at least closed monoidal, coming equipped with a \textit{weak dual}.
\begin{definition}
  Let $(\cat{C},\otimes,i)$ be a closed monoidal category with internal-hom $[-,-]$.
  We say that the object $a^*:=[a,i]$ is the \textit{weak dual} to the object $a$.
\end{definition}
A $*$-autonomous category can be thought of roughly as a closed monoidal category in which the weak dual $(-)^{*} $ satisfies $(-)^{**} \cong (-)$.
In other words, a $*$-autonomous category is a closed monoidal category with an involutive weak dual.
As we will see, there exist choices of $\cat{C}$ for which $\StProf(\cat{C})$ is not $*$-autonomous, and for good physical reason, the presence of involutivity implies supermap decomposition theorems which \textit{cannot} hold in general. 

Concretely, in the $\mathsf{Caus}$-construction the dual of an object $(A,c_A)$ is given by $(A^*,c_A^*)$, where $c_A^*$ is the set of effects that are causal on all states in $c_A$.
In this case checking involutivity amounts to checking $(A,c_a)^{**} = (A^{**},c_A^{**}) \cong (A,c_a)$, which follows because the underlying category is taken to be compact closed ($A\cong A^{**}$) and the sets $c_A$ are taken to be closed.
This reliance on the underlying compact closure of the category at hand makes it unclear whether we ought to expect any negation we can define on $\StProf(\cat{C})$ to be involutive, not least because, unlike the $\mathsf{Caus}$-construction we do not have the underlying structure of a compact closed category (which in linear algebra is the signature of finite dimensionality) to begin with.

\subsection{Supermap Decomposition from Involutivity}

There is a more concrete reason than dimensionality why we cannot expect $\StProf(\cat{C})$ to be $*$-autonomous. The exchange of states and effects cannot in general be involutive in this higher-order setting due to the lack of a general decomposition theorem for categorical supermaps. For instance, consider the object $y_{\vb*{a}}$ of $\StProf(\cat{C})$, that is the space of optics with input $(a,a')$.
Its weak dual (for the closed monoidal operation $\otimes_\cat{C}$) is given by $\suphom{a}{a'}$. 

\begin{proposition}\label{prop:dualoptics}
  The weak dual of $y_{\vb*{a}}$ is $\suphom{a}{a'}$.
\end{proposition}
\begin{proof}
  $[y_{\vb*{a}},1] \cong \int_{\vb*{x}} \set\left( y_{\vb*{a}}\vb*{x}, y_{\vb*{i}} (\vb*{x}\otimes -) \right) \cong y_{\vb*{i}}(\vb*{a}\otimes -) \cong \suphom{a}{a'}$.
\end{proof}

If we interpret $y_{\vb*{a}}$ as the hole $(a,a')$ in the category $\cat{C}$ then we have that the weak dual to a hole is the space of maps $\suphom{a}{a'}$ with which we could fill the hole.

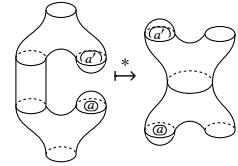
\begin{xlrbox}{optic_shape}
	\begin{tikzpicture}[baseline=(current bounding box.center),{every node/.style}={scale=0.8}]
		\node[pants,bot,top] (p) {};
		\node[tube,bot,anchor=top] (t) at (p.leftleg) {};
		\node[copants,bot,anchor=leftleg] (c) at (t.bot) {};
		\node[bowl] at (p.rightleg) {};
		\node[cobowl] at (c.rightleg) {};
		\begin{pgfonlayer}{nodelayer}
			\node[representable] at (p.rightleg) {$a'$};
			\node[representable] at (c.rightleg) {$a$};
		\end{pgfonlayer}
	\end{tikzpicture}
\end{xlrbox}

\begin{xlrbox}{supermap_basic_shape}
	\begin{tikzpicture}[baseline=(current bounding box.center),{every node/.style}={scale=0.8}]
		\node[copants,bot,widebelt] (c) {};
		\node[pants,bot,widebelt,anchor=belt] (p) at (c.belt) {};
		\node[bowl] at (p.leftleg) {};
		\node[cobowl] at (c.leftleg) {};
		\node[end] at (c.rightleg) {};
		\begin{pgfonlayer}{nodelayer}
			\node[representable] at (c.leftleg) {$a'$};
			\node[representable] at (p.leftleg) {$a$};
		\end{pgfonlayer}
	\end{tikzpicture}
\end{xlrbox}

\begin{equation*}
	\xusebox{optic_shape} \overset{*}{\mapsto} \xusebox{supermap_basic_shape}
\end{equation*}

This gives part of a duality between holes in the category $\cat{C}$ and the maps that fill those holes, but it is not clear that the dual of $\suphom{a}{a'}$ should be $y_{\vb*{a}}$ again.
It seems more reasonable to expect that the dual to $\suphom{a}{a'}$ is the space of all supermaps on the hole $(a,a')$, and there is no reason a priori that this should be $y_{\vb*{a}}$, for instance if there exist supermaps which are not of the form of an optic.
As a result it is not generally the case that the weak dual of $\suphom{a}{a'}$ is $y_{\vb*{a}}$ and this precludes any hope of $\StProf(\cat{C})$ being a $*$-autonomous category in general.

This intuitive idea can be made into a formal statement, which establishes the behaviour of double duals as the key concept in the study of decomposition theorems for categorical supermaps. 

\begin{proposition}\label{prop:decomp_thms}
	A symmetric monoidal category $\cat{C}$ has a $1$-arity decomposition theorem if and only if
	\[\suphom{a}{a'}^* \cong y_{\vb*{a}}, \ \ \textrm{or equivalently}, \ \ y_{\vb*{a}}^{**} \cong y_{\vb*{a}}. \] 
	Furthermore, $\cat{C}$ has an $n$-arity supermap decomposition theorem if and only if 
	\begin{equation*}
		{( \seq_i {y_{\vb*{a}_i }^{*} })}^{*} \cong \seq_i y_{\vb*{a}_i}. 
	\end{equation*}
\end{proposition}
\begin{proof}
	Given in Appendix \ref{sec:decomp_thms_proof}.
\end{proof}

In other words to ask that the dual to $\suphom{a}{a'}$ is given by optics is to ask that the supermaps between the primitive types $\suphom{a}{a'}$ and $\suphom{b}{b'}$ are equivalent to optics.
It follows that the category $\StProf(\cptp)$ has additional nice categorical properties in terms of weak duals, precisely stemming from its supermap decomposition theorems.
In short we can now express supermap decompositions for $\cptp$ fully with the basic categorical condition ${( \seq_i {y_{\vb*{a}_i }^{*} })}^{*} \cong \seq_i y_{\vb*{a}_i}$ involving weak duals of Yoneda embeddings and monoidal products of profunctors. 

For an example of a category where these dualities do not hold consider $\mathsf{Unitaries}$.
In \cite{wilson_polycategories} it is demonstrated that there exist morphisms of strong profunctors on $\suphom{a}{a'}$ which do not admit a representation in terms of even a linear map, let alone an optic. Consequently, $\StProf(\mathsf{Unitaries})$ cannot be $*$-autonomous. 

\subsection{Tensorial Logic}
We have seen that the kind of dual or negation present in $\StProf(\cat{C})$ is not strong enough to give $*$-autonomous structure, and in fact, we have seen that this is a feature not a bug. It is still important on the other hand to extract logical structure that \textit{is} present for fully abstract categorical supermaps. 

To this end, whilst $*$-autonomous categories give models of linear logic, there is another more general and primitive logic known as \textit{tensorial logic} \cite{mellies_tensorlogic,mellies_continuation,mellies_negation} which weakens the involutivity of $(-)^*$, replacing it instead with a tensorial strength.
The models of tensorial logic are known as dialogue categories.

\begin{definition}
	A \textit{dialogue} category is a symmetric monoidal category $(\cat{C},\otimes,i)$ with a negation functor $(-)^*:\opcat{\cat{C}}\morph{}\cat{C}$ with natural isomorphisms $\cat{C}(a\otimes b,c^*) \cong \cat{C}(a,(b\otimes c)^*)$ satisfying a certain commutative diagram \cite{mellies_tensorlogic}.
\end{definition}

Any closed monoidal category becomes straightforwardly a dialogue category where the negation functor is given by taking weak dual objects.
The endofunctor $[[-,i],i]:\cat{C}\morph{}\cat{C}$ (which would have needed to be isomorphic to the identity functor for $*$-autonomy) is well-known in computer science as the \textit{continuation monad}.
One way of seeing that this functor is monadic is to note that it arises from the adjunction between $[-,i]$ and $\opcat{[-,i]}$.

In any dialogue category the negation functor has a parametrised strength $(x\otimes b^*)^* \otimes c \morph{} (x \otimes (b\otimes c)^*)^*$.
Taking the parameter $x$ to be the monoidal unit $i$ returns a more usual strength $b^{**}\otimes c \morph{} (b\otimes c)^{**}$, while taking $x$ to be $a^*$ returns the distributivity between $\otimes$ and $\amp$ of linear logic,
\begin{equation*}
	(a\amp b)\otimes c = (a^* \otimes b^*)^*\otimes c \morph{} (a^* \otimes(b\otimes c)^*)^* = a\amp (b\otimes c).
\end{equation*}
It should be noted that while $-\amp- := (-^*\otimes -^*)^*$ is a functor $\cat{C}\times\cat{C}\morph{}\cat{C}$ it is in general not associative or unital and does not become a tensor product of $\cat{C}$.

Nevertheless, even with weaker tensorial logic we can use this new connective $\amp$ to generate further spaces of maps.
For instance suppose that $\cat{C}$ has a 1-arity supermap decomposition theorem so that $\suphom{a}{a'}^*\cong y_{\vb*{a}}$ for all $\vb*{a}$.
Then we can calculate,
\begin{align*}
	& \suphom{a}{a'}\amp \suphom{b}{b'} \\
	& = (\suphom{a}{a'}^* \otimes_\cat{C} \suphom{b}{b'}^*)^* \\
	& \cong (y_{\vb*{a}} \otimes_\cat{C} y_{\vb*{b}}) ^* \cong y_{\vb*{a}\otimes \vb*{b}}^* \\
	& \cong \suphom{a\otimes b}{a'\otimes b'}
\end{align*}
So the par $\amp$ sends the space of maps on party $\vb*{a}$ and the space of maps on $\vb*{b}$ to the space of all bipartite channels $a\otimes b\morph{}a'\otimes b'$ (plus environmental extension).
The linear distributors then give inclusions between spaces of separable and arbitrary channels.

Conversely, even without any supermap decomposition theorems, the par of two holes is the space of multi-partite indefinitely causally ordered supermaps on those holes.
\begin{align*}
	y_{\vb*{a}} \amp y_{\vb*{b}} & = (y_{\vb*{a}}^* \otimes_\cat{C} y_{\vb*{a}}^*)^* \\
	& \cong (\suphom{a}{a'}\otimes_\cat{C}\suphom{b}{b'})^*
\end{align*}
To see this as a space of supermaps note the following.
\begin{align*}
	& \begin{multlined}[t] \StProf(\cat{C}) \big( \suphom{a}{a'}\otimes_\cat{C}\suphom{b}{b'}, \\
		\suphom{c}{c'} \big) \end{multlined} \\
	& \cong \StProf(\cat{C}) \big( y_{\vb*{c}}, (\suphom{a}{a'}\otimes_\cat{C}\suphom{b}{b'})^* \big) \\
	& \cong (\suphom{a}{a'}\otimes_\cat{C}\suphom{b}{b'})^*(\vb*{c})
\end{align*} 

\section{Conclusion and Future Work}
We have shown how to model supermaps, black-box holes, using the category of strong endoprofunctors. 
This category is rich enough to encompass and organise the abstract study of supermaps with indefinite causal order, supermaps with abstract but definite causal structure, \textit{and} through the Yoneda embedding the supermaps which concretely decompose into networks quotiented by equivalence up to the sliding of boxes.

Our approach builds a bridge between the diagrammatic study of quantum supermaps \cite{wilson_locality, wilson_polycategories} and the categorical approach to the study of concrete diagrams with holes in terms of profunctor and coend optics \cite{pastro_street,clarke_profunctor,roman_coend,roman_comb,roman_optics,roman_thesis,riley_optics,earnshaw,hefford_coend,boisseau_optics}.
This significantly tames the study of supermaps on monoidal categories by providing new logical connectives and a new toolbox, the coend calculus, which is well suited for reasoning in these abstract settings.

The applications to quantum foundations are clear: we have a model for supermaps on separable and non-separable Hilbert spaces which has working tensor products along with a (non-involutive) negation, giving a model for tensorial logic \cite{mellies_tensorlogic,mellies_continuation,mellies_negation}. The next steps towards application hinge on the problem of characterisation of the produced categorical definitions, hopefully in terms of linear maps and other nice concrete properties.
Similar considerations hold for the study of indefinite causal structure in OPTs \cite{chiribella_purification, chiribella_informational}.
Beyond this, there is the question of the study of supermaps on more elaborate kinds of theories and spaces, which optics and strong profunctors are well equipped to handle, being generalisable to actegories \cite{clarke_profunctor}, Freyd-categories \cite{hefford_optics} and even double categories \cite{capucci_seeing}.
In particular, could this approach be further developed, to categorify the study of super-super maps and their iterations \cite{kissinger_caus, Bisio_2019, hoffreumon2022projective, apadula2022nosignalling, SimmonsKissinger2022, Wilson_causal, wilson2023mathematical, Perinotti_2017}? 

More open ended is the possibility of importing ideas from the foundations of physics into theoretical computer science and applied category theory.
Indeed coming from the other side of this bridge, there is a vast and rich literature from the field of quantum foundations on the kinds of non-causally ordered boxes-with-holes which can be constructed \cite{baumeler_logically, wechs, oreshkov, chiribella_switch}. It is natural to wonder whether (a) black-box holes and (b) indefinite causal orders could have a role to play in other areas of computer science.

In retrospect, the category-theoretic interpretation for supermaps we have presented is quite natural: given a category $\cat{D}$, $\opcat{\copsh{\cat{D}}}$ is, up to size issues, the free completion.
Taking this free completion adds in more abstract, less concrete spaces and given that the study of supermaps is all about the study of things which are like diagrams with holes but more abstract and less concrete, it is very pleasing to see the study of supermaps identified categorically as the study of copresheaves over the category $\opt(\cat{C})$ of concrete diagrams with holes. 

\begin{acks}
	We would like to thank Mario Román for discussions about the monoidal products of optics, in which he pointed out their produoidal structure, and for contributing writing and discussions to an early version of this project.
	We would also like to thank Marc de Visme for pointing us towards a useful lemma (Lemma 8 of \cite{carette_delayedtrace}) for the proof of Theorem \ref{thm:cptp_equiv} and Cole Comfort for many helpful discussions.
	The internal string diagrams were made using a package heavily inspired by that of \cite{bartlett_extended,bartlett_modular}.
	This work has been partially funded by the French National Research Agency (ANR) within the framework of ``Plan France 2030'', under the research projects EPIQ ANR-22-PETQ-0007 and HQI-R\&D ANR-22-PNCQ-0002. MW was funded by the Engineering and Physical Sciences Research Council [grant number EP/W524335/1].
\end{acks}

\bibliographystyle{ACM-Reference-Format}
\bibliography{bibliography}

\appendix

\section{Proofs}
\subsection{Proof of Theorem \ref{thm:cptp_equiv}}\label{sec:cptp_proof}
\begin{theorem*}(Restated for clarity).
	There is a symmetric monoidal equivalence of categories $\comb(\cptp)\cong\opt(\cptp)$.
\end{theorem*}
\begin{xlrbox}{cptp_proof1}
	\begin{tikzpicture}[baseline=(current bounding box.center)]
		\node[pants,bot,top] (p) {};
		\node[tube,bot,anchor=top] (t) at (p.leftleg) {};
		\node[copants,bot,anchor=leftleg] (c) at (t.bot) {};
		\node[end] at (c.rightleg) {};
		\begin{pgfonlayer}{nodelayer}
			\node[label] (g) at (p.center) {$g'$};
			\node[label] (f) at (c.center) {$f'$};
		\end{pgfonlayer}
		\begin{pgfonlayer}{strings}
			\draw (g) to[out=-135,in=90] (p.leftleg) {};
			\draw (g) to[out=-45,in=90] (p.rightleg) {};
			\draw (f) to[out=135,in=-90] (c.leftleg) {};
			\draw (f) to[out=45,in=-90] (c.rightleg) {};
			\draw (g) to (p.belt) {};
			\draw (f) to (c.belt) {};
			\draw (t.top) to (t.bot) {};
		\end{pgfonlayer}
	\end{tikzpicture}
\end{xlrbox}

\begin{xlrbox}{cptp_proof2}
	\begin{tikzpicture}[baseline=(current bounding box.center)]
		\node[pants,bot,top,widebelt] (p) {};
		\node[tube,bot,anchor=top] (t) at (p.leftleg) {};
		\node[copants,bot,widebelt,anchor=leftleg] (c) at (t.bot) {};
		\node[end] at (c.rightleg) {};
		\begin{pgfonlayer}{nodelayer}
			\node[label] (g) at (p.center) {$q'$};
			\node[label] (f) at ([yshift=-2]c.center) {$p'$};
			\node[smalldiscarding,scale=0.4] (d1) at ([xshift=-8,yshift=6]f) {};
			\node[smalldiscarding,scale=0.4] (d2) at ([xshift=-8,yshift=6]g) {};
		\end{pgfonlayer}
		\begin{pgfonlayer}{strings}
			\draw (g) to[out=-135,in=90] (p.leftleg) {};
			\draw (g) to[out=-45,in=90] (p.rightleg) {};
			\draw (f) to[out=90,in=-90] (c.leftleg) {};
			\draw (f) to[out=45,in=-90] (c.rightleg) {};
			\draw (f) to[out=180,in=-90] (d1.west) {};
			\draw (g) to[out=180,in=-90] (d2.west) {};
			\draw (g) to (p.belt) {};
			\draw (f) to (c.belt) {};
			\draw (t.top) to (t.bot) {};
		\end{pgfonlayer}
	\end{tikzpicture}
\end{xlrbox}

\begin{xlrbox}{cptp_proof3}
	\begin{tikzpicture}[baseline=(current bounding box.center)]
		\node[pants,bot,top,widebelt] (p) {};
		\node[tube,bot,anchor=top] (t) at (p.leftleg) {};
		\node[copants,bot,widebelt,anchor=leftleg] (c) at (t.bot) {};
		\node[end] at (c.rightleg) {};
		\begin{pgfonlayer}{nodelayer}
			\node[label] (g) at (p.center) {$q'$};
			\node[label] (f) at ([yshift=-2]c.center) {$p$};
			\node[label,inner sep=0.06cm] (v) at ([xshift=-8,yshift=6]f) {$v$};
			\node[smalldiscarding,scale=0.4] (d1) at ([xshift=-8,yshift=5]v) {};
			\node[smalldiscarding,scale=0.4] (d2) at ([xshift=-8,yshift=6]g) {};
		\end{pgfonlayer}
		\begin{pgfonlayer}{strings}
			\draw (g) to[out=-135,in=90] (p.leftleg) {};
			\draw (g) to[out=-45,in=90] (p.rightleg) {};
			\draw (f) to[out=90,in=0] (v) {};
			\draw (f) to[out=180,in=-90] (v) {};
			\draw (f) to[out=45,in=-90] (c.rightleg) {};
			\draw (v) to[out=180,in=-90] (d1.west) {};
			\draw (g) to[out=180,in=-90] (d2.west) {};
			\draw (g) to (p.belt) {};
			\draw (f) to (c.belt) {};
			\draw (v) to[out=90,in=-90] (c.leftleg) to[out=90,in=-90] (t.top) {};
		\end{pgfonlayer}
	\end{tikzpicture}
\end{xlrbox}

\begin{xlrbox}{cptp_proof4}
	\begin{tikzpicture}[baseline=(current bounding box.center)]
		\node[pants,bot,top,widebelt] (p) {};
		\node[tube,bot,anchor=top] (t) at (p.leftleg) {};
		\node[copants,bot,widebelt,anchor=leftleg] (c) at (t.bot) {};
		\node[end] at (c.rightleg) {};
		\begin{pgfonlayer}{nodelayer}
			\node[label] (g) at ([yshift=2]p.center) {$q'$};
			\node[label] (f) at ([yshift=-2]c.center) {$p$};
			\node[label,inner sep=0.06cm] (v) at ([yshift=4]p.leftleg) {$v$};
			\node[smalldiscarding,scale=0.4] (d1) at ([xshift=1,yshift=6]v) {};
			\node[smalldiscarding,scale=0.4] (d2) at ([xshift=-8,yshift=6]g) {};
		\end{pgfonlayer}
		\begin{pgfonlayer}{strings}
			\draw (g) to[out=-45,in=90] (p.rightleg) {};
			\draw (f) to[out=45,in=-90] (c.rightleg) {};
			\draw (f) to[out=90,in=-90] ([xshift=2]c.leftleg) to ([xshift=2]p.leftleg) to (v) {};
			\draw (f) to[out=180,in=-90] ([xshift=-2]c.leftleg) to ([xshift=-2]p.leftleg) to (v) {};
			\draw (v) to[out=90,in=-90] (d1.west) {};
			\draw (g) to[out=180,in=-90] (d2.west) {};
			\draw (g) to (p.belt) {};
			\draw (f) to (c.belt) {};
			\draw (v) to[out=45,in=-135] (g) {};
		\end{pgfonlayer}
	\end{tikzpicture}
\end{xlrbox}

\begin{xlrbox}{cptp_proof5}
	\begin{tikzpicture}[baseline=(current bounding box.center)]
		\node[pants,bot,top,widebelt] (p) {};
		\node[tube,bot,anchor=top] (t) at (p.leftleg) {};
		\node[copants,bot,widebelt,anchor=leftleg] (c) at (t.bot) {};
		\node[end] at (c.rightleg) {};
		\begin{pgfonlayer}{nodelayer}
			\node[label] (g) at ([yshift=2]p.center) {$q'$};
			\node[label] (f) at ([yshift=-2]c.center) {$p$};
			\node[label,inner sep=0.06cm] (v) at ([yshift=4]p.leftleg) {$v$};
			\node[label,inner sep=0.06cm] (pi) at ([xshift=-8,yshift=6]f) {$\pi$};
			\node[smalldiscarding,scale=0.4] (d1) at ([xshift=1,yshift=6]v) {};
			\node[smalldiscarding,scale=0.4] (d2) at ([xshift=-8,yshift=6]g) {};
		\end{pgfonlayer}
		\begin{pgfonlayer}{strings}
			\draw (g) to[out=-45,in=90] (p.rightleg) {};
			\draw (f) to[out=45,in=-90] (c.rightleg) {};
			\draw (pi) to[out=135,in=-90] ([xshift=2]c.leftleg) to ([xshift=2]p.leftleg) to (v) {};
			\draw (pi) to[out=180,in=-90] ([xshift=-2]c.leftleg) to ([xshift=-2]p.leftleg) to (v) {};
			\draw (v) to[out=90,in=-90] (d1.west) {};
			\draw (g) to[out=180,in=-90] (d2.west) {};
			\draw (g) to (p.belt) {};
			\draw (f) to (c.belt) {};
			\draw (v) to[out=45,in=-135] (g) {};
			\draw (f) to[out=90,in=0] (pi) {};
			\draw (f) to[out=180,in=-90] (pi) {};
		\end{pgfonlayer}
	\end{tikzpicture}
\end{xlrbox}

\begin{xlrbox}{cptp_proof6}
	\begin{tikzpicture}[baseline=(current bounding box.center)]
		\node[pants,top,widebelt] (p) {};
		\node[tube,bot,anchor=top] (t) at (p.leftleg) {};
		\node[copants,bot,widebelt,anchor=leftleg] (c) at (t.bot) {};
		\node[end] at (c.rightleg) {};
		\node[end_dot] at (p.rightleg) {};
		\node[end_dot] (e) at ([yshift=-6]p.leftleg) {};
		\begin{pgfonlayer}{nodelayer}
			\node[label] (g) at ([yshift=2]p.center) {$q'$};
			\node[label] (f) at ([yshift=-2]c.center) {$p$};
			\node[label,inner sep=0.06cm] (v) at ([yshift=4]p.leftleg) {$v$};
			\node[label,inner sep=0.06cm] (pi) at (e) {$\pi$};
			\node[smalldiscarding,scale=0.4] (d1) at ([xshift=1,yshift=6]v) {};
			\node[smalldiscarding,scale=0.4] (d2) at ([xshift=-8,yshift=6]g) {};
		\end{pgfonlayer}
		\begin{pgfonlayer}{strings}
			\draw (g) to[out=-45,in=90] (p.rightleg) {};
			\draw (f) to[out=45,in=-90] (c.rightleg) {};
			\draw (f) to[out=135,in=-90] ([xshift=2]c.leftleg) to ([xshift=2]e) to (pi) {};
			\draw (f) to[out=180,in=-90] ([xshift=-2]c.leftleg) to ([xshift=-2]e) to (pi) {};
			\draw (v) to[out=90,in=-90] (d1.west) {};
			\draw (g) to[out=180,in=-90] (d2.west) {};
			\draw (g) to (p.belt) {};
			\draw (f) to (c.belt) {};
			\draw (v) to[out=45,in=-135] (g) {};
			\draw (pi) to[out=135,in=-135] (v) {};
			\draw (pi) to[out=45,in=-45] (v) {};
		\end{pgfonlayer}
	\end{tikzpicture}
\end{xlrbox}

\begin{xlrbox}{cptp_proof7}
	\begin{tikzpicture}[baseline=(current bounding box.center)]
		\node[pants,bot,top,widebelt] (p) {};
		\node[tube,bot,anchor=top] (t) at (p.leftleg) {};
		\node[copants,bot,widebelt,anchor=leftleg] (c) at (t.bot) {};
		\node[end] at (c.rightleg) {};
		\begin{pgfonlayer}{nodelayer}
			\node[label] (g) at ([yshift=2]p.center) {$q$};
			\node[label] (f) at ([yshift=-2]c.center) {$p$};
			\node[label,inner sep=0.06cm] (v) at ([yshift=4]p.leftleg) {$\pi$};
			\node[smalldiscarding,scale=0.4] (d1) at ([xshift=1,yshift=6]v) {};
			\node[smalldiscarding,scale=0.4] (d2) at ([xshift=-8,yshift=6]g) {};
		\end{pgfonlayer}
		\begin{pgfonlayer}{strings}
			\draw (g) to[out=-45,in=90] (p.rightleg) {};
			\draw (f) to[out=45,in=-90] (c.rightleg) {};
			\draw (f) to[out=90,in=-90] ([xshift=2]c.leftleg) to ([xshift=2]p.leftleg) to (v) {};
			\draw (f) to[out=180,in=-90] ([xshift=-2]c.leftleg) to ([xshift=-2]p.leftleg) to (v) {};
			\draw (v) to[out=90,in=-90] (d1.west) {};
			\draw (g) to[out=180,in=-90] (d2.west) {};
			\draw (g) to (p.belt) {};
			\draw (f) to (c.belt) {};
			\draw (v) to[out=45,in=-135] (g) {};
		\end{pgfonlayer}
	\end{tikzpicture}
\end{xlrbox}

\begin{xlrbox}{cptp_proof8}
	\begin{tikzpicture}[baseline=(current bounding box.center)]
		\node[pants,bot,top,widebelt] (p) {};
		\node[tube,bot,anchor=top] (t) at (p.leftleg) {};
		\node[copants,bot,widebelt,anchor=leftleg] (c) at (t.bot) {};
		\node[end] at (c.rightleg) {};
		\begin{pgfonlayer}{nodelayer}
			\node[label] (g) at (p.center) {$q$};
			\node[label] (f) at ([yshift=-2]c.center) {$p$};
			\node[label,inner sep=0.06cm] (v) at ([xshift=-8,yshift=6]f) {$\pi$};
			\node[smalldiscarding,scale=0.4] (d1) at ([xshift=-8,yshift=5]v) {};
			\node[smalldiscarding,scale=0.4] (d2) at ([xshift=-8,yshift=6]g) {};
		\end{pgfonlayer}
		\begin{pgfonlayer}{strings}
			\draw (g) to[out=-135,in=90] (p.leftleg) {};
			\draw (g) to[out=-45,in=90] (p.rightleg) {};
			\draw (f) to[out=90,in=0] (v) {};
			\draw (f) to[out=180,in=-90] (v) {};
			\draw (f) to[out=45,in=-90] (c.rightleg) {};
			\draw (v) to[out=180,in=-90] (d1.west) {};
			\draw (g) to[out=180,in=-90] (d2.west) {};
			\draw (g) to (p.belt) {};
			\draw (f) to (c.belt) {};
			\draw (v) to[out=90,in=-90] (c.leftleg) to[out=90,in=-90] (t.top) {};
		\end{pgfonlayer}
	\end{tikzpicture}
\end{xlrbox}

\begin{xlrbox}{cptp_proof9}
	\begin{tikzpicture}[baseline=(current bounding box.center)]
		\node[pants,bot,top,widebelt] (p) {};
		\node[tube,bot,anchor=top] (t) at (p.leftleg) {};
		\node[copants,bot,widebelt,anchor=leftleg] (c) at (t.bot) {};
		\node[end] at (c.rightleg) {};
		\begin{pgfonlayer}{nodelayer}
			\node[label] (g) at (p.center) {$q$};
			\node[label] (f) at ([yshift=-2]c.center) {$p$};
			\node[smalldiscarding,scale=0.4] (d1) at ([xshift=-8,yshift=6]f) {};
			\node[smalldiscarding,scale=0.4] (d2) at ([xshift=-8,yshift=6]g) {};
		\end{pgfonlayer}
		\begin{pgfonlayer}{strings}
			\draw (g) to[out=-135,in=90] (p.leftleg) {};
			\draw (g) to[out=-45,in=90] (p.rightleg) {};
			\draw (f) to[out=90,in=-90] (c.leftleg) {};
			\draw (f) to[out=45,in=-90] (c.rightleg) {};
			\draw (f) to[out=180,in=-90] (d1.west) {};
			\draw (g) to[out=180,in=-90] (d2.west) {};
			\draw (g) to (p.belt) {};
			\draw (f) to (c.belt) {};
			\draw (t.top) to (t.bot) {};
		\end{pgfonlayer}
	\end{tikzpicture}
\end{xlrbox}

\begin{xlrbox}{cptp_proof10}
	\begin{tikzpicture}[baseline=(current bounding box.center)]
		\node[pants,bot,top] (p) {};
		\node[tube,bot,anchor=top] (t) at (p.leftleg) {};
		\node[copants,bot,anchor=leftleg] (c) at (t.bot) {};
		\node[end] at (c.rightleg) {};
		\begin{pgfonlayer}{nodelayer}
			\node[label] (g) at (p.center) {$g$};
			\node[label] (f) at (c.center) {$f$};
		\end{pgfonlayer}
		\begin{pgfonlayer}{strings}
			\draw (g) to[out=-135,in=90] (p.leftleg) {};
			\draw (g) to[out=-45,in=90] (p.rightleg) {};
			\draw (f) to[out=135,in=-90] (c.leftleg) {};
			\draw (f) to[out=45,in=-90] (c.rightleg) {};
			\draw (g) to (p.belt) {};
			\draw (f) to (c.belt) {};
			\draw (t.top) to (t.bot) {};
		\end{pgfonlayer}
	\end{tikzpicture}
\end{xlrbox}

\begin{proof}
	Suppose that $(f,g)\sim(f',g')$ in $\comb(\cptp)$.
	The maps $f,f',g$ and $g'$ each possess purifications say $p,p',q$ and $q'$ respectively.
	Since $(f,g)\sim(f',g')$, we know that the combs are equal on the swap.

	\begin{equation}\label{eq:cptp_proof1}
		\tikzfigscale{1}{figs/cptp_proof1}
	\end{equation}

	Upon discarding the left-hand output and noting that $q$ and $q'$ are causal, we find that 

	\begin{equation}\label{eq:cptp_proof2}
		\tikzfigscale{1}{figs/cptp_proof2}
	\end{equation}

	By the Stinespring dilation theorem, this means that either there exists an isometry $v$ on the dilating space such that $(v\otimes 1)p=p'$ or an isometry $v'$ such that $p=(v'\otimes 1)p'$.
	Without loss of generality we assume the former.

	Thus we now know that the following two maps are equal.

	\begin{equation}\label{eq:cptp_proof3}
		\tikzfigscale{1}{figs/cptp_proof3}
	\end{equation}

	We can invoke Lemma 8 of \cite{carette_delayedtrace} which demonstrates that all pure maps of $\cptp$ have \textit{shadows}.
	Since $p$ is pure, this means there exists an idempotent CPTP map $\pi$ such that 

	\begin{equation}\label{eq:cptp_proof4}
		\tikzfigscale{1}{figs/cptp_proof4} \qquad \text{and} \qquad \tikzfigscale{1}{figs/cptp_proof5}
	\end{equation}

	Finally we can demonstrate that $(f,g)\sim(f',g')$ in $\opt(\cptp)$.

	\begin{equation*}
		\xusebox{cptp_proof1} = \xusebox{cptp_proof2} = \xusebox{cptp_proof3} \sim \xusebox{cptp_proof4}
	\end{equation*}
	\begin{equation*}
		= \xusebox{cptp_proof5} \sim \xusebox{cptp_proof6} = \xusebox{cptp_proof7} \sim \xusebox{cptp_proof8} 
	\end{equation*}
	\begin{equation*}
		= \xusebox{cptp_proof9} = \xusebox{cptp_proof10}
	\end{equation*}
\end{proof}

\subsection{Proof of Theorem \ref{thm:cptp_equiv2}}\label{sec:cptp_proof2}
\begin{theorem*}(Restated for clarity).
	There is an isomorphism of produoidal categories $\comb(\cptp)\cong\opt(\cptp)$.
\end{theorem*}
\begin{proof}
	Since the categories $\comb(\cat{C})$ and $\opt(\cat{C})$ are quite so similar, the main thing we need to prove is that the spaces of $n$-holed combs and $n$-holed optics are isomorphic for $\cptp$.
	This will then establish that the promonoidal structures are isomorphic.

	We proceed by induction.
	We already know from Theorem \ref{thm:cptp_equiv} that the one-holed combs and optics are isomorphic.
	This establishes the base case.

	Now suppose that equivalence of $(n-1)$-holed combs implies equivalence of $(n-1)$-holed optics.
	Consider two equivalent $n$-holed combs $(f_0,f_1,\dots,f_n)\sim (f'_0,f'_1,\dots,f'_n)$.
	By essentially the same argument as for Theorem \ref{thm:cptp_equiv} we will reduce this down to an equivalence of $(n-1)$-combs, thereby establishing the result.

	Purifying each $f_i$ and $f_i'$ to $p_i$ and $p_i'$, and inserting swaps into the $n$-combs gives an equality analogous to equation \eqref{eq:cptp_proof1}.
	Upon discarding the left-most $n$ outputs one finds an analogous equation to \eqref{eq:cptp_proof2} showing that the bottom maps of the combs $p_0$ and $p_0'$ are equal up to discarding their left-hand output.
	By Stinespring, they are therefore equal up to an isometry $v$ on the environment and we retrieve an equation analogous to \eqref{eq:cptp_proof3}.

	Lemma 8 of \cite{carette_delayedtrace} allows us to do a similar trick to remove the map $p_0$ at the bottom of the combs in exchange for a CPTP map $\pi$, much like \eqref{eq:cptp_proof4}.
	Since $\cptp$ embeds faithfully into the category of all CP maps which is compact closed, the analogous equation to \eqref{eq:cptp_proof4} implies that the following $(n-1)$-combs are equivalent 
	\begin{equation}\label{eq:cptp2_proof}
		(p_1\circ (\pi\otimes 1), p_2,\dots,p_n) \sim (p'_1\circ (v\pi\otimes 1), p'_2,\dots,p'_n)
	\end{equation}

	A similar argument to the final part of the proof of Theorem \ref{thm:cptp_equiv} completes the proof.
	Starting with the optic $(f'_0,f'_1,\dots,f'_n)$ one rewrites this using the purifications $(p'_0, p'_1,\dots,p'_n)$.
	One can then introduce the maps $v$ and $\pi$ thereby converting $p_0'$ into $p_0$.
	Upon sliding $v$ and $\pi$ using the coend equivalence relation, one obtains the optic $(p'_1\circ (v\pi\otimes 1), p'_2,\dots,p'_n)$ which is equivalent to the optic $(p_1\circ (\pi\otimes 1), p_2,\dots,p_n)$ by equation \eqref{eq:cptp2_proof} and the inductive hypothesis.
\end{proof}

\subsection{Proof of Theorem \ref{thm:separable_supermaps}}\label{sec:separable_supermaps_proof}
\begin{theorem*}(Restated for clarity).
	On any symmetric monoidal category $\cat{C}$, the strong natural transformations of type 
	\begin{equation*}
		\eta: \modtensor^{\hspace{-0.7em}n}_{\hspace{-0.7em}i=1} \suphom{a_i}{a_i'} \rightarrow \suphom{b}{b'}
	\end{equation*}
	are the multi-partite locally-applicable transformations of type $\eta:\vb*{a}_1,\dots,\vb*{a}_n\morph{}\vb*{b}$.
\end{theorem*}
\begin{proof}
	For simplicity we consider the case $n=2$ as it is straightforward to generalise to larger $n$.
	Let $P_1$, $P_2$ and $R$ be strong profunctors.
	We have the following natural isomorphism arising from the adjunction between the representable and corepresentable embeddings of the tensor product of $\opt(\cat{C})$ in $\Prof$.
	\begin{align*}
		& \StProf(\cat{C})( P_1 \otimes_{\cat{C}} P_2 , R) \\
		& \cong \StProf(\cat{C}) \left( \int^{\vb*{a}\vb*{b}}  P_1\vb*{a} \times P_2\vb*{b} \times \opt(\cat{C}) (\vb*{a} \otimes \vb*{b},  -) , R(-) \right) \\
		& \cong \int_{\vb*{c}} \set \left( \int^{\vb*{a}\vb*{b}}  P_1\vb*{a} \times P_2\vb*{b} \times \opt(\cat{C}) (\vb*{a} \otimes \vb*{b},  \vb*{c}) , R\vb*{c} \right) \\
		& \cong \int_{\vb*{a}\vb*{b}\vb*{c}} \set \big( P_1\vb*{a} \times P_2\vb*{b} \times \opt(\cat{C}) (\vb*{a} \otimes \vb*{b},  \vb*{c}) , R\vb*{c} \big) \\
		& \cong \int_{\vb*{a}\vb*{b}\vb*{c}} \set \big( P_1\vb*{a} \times P_2\vb*{b} , \set( \opt(\cat{C})(\vb*{a} \otimes \vb*{b},  \vb*{c}), R\vb*{c}) \big) \\
		& \cong \int_{\vb*{a}\vb*{b}} \set \left( P_1\vb*{a} \times P_2\vb*{b} , \int_{\vb*{c}} \set( \opt(\cat{C}) (\vb*{a} \otimes \vb*{b}, \vb*{c}), R\vb*{c} ) \right) \\    
		& \cong \int_{\vb*{a}\vb*{b}} \set ( P_1\vb*{a} \times P_2\vb*{b} , R(\vb*{a} \otimes \vb*{b})) \\
		& \cong \copsh{\opt(\cat{C}) \times \opt(\cat{C})}(P_1(-) \times P_2(\bl) , R(- \otimes \bl))
	\end{align*}
	An element of the above set is an assignment of a family of functions 
	\begin{equation*}
		\eta_{(x_1 , x_1')(x_2 , x_2')} : P_1(x_1 , x_1') \times P_2(x_2 , x_2') \rightarrow R(x_1 \otimes x_2 , x_1' \otimes x_2')
	\end{equation*} 
	to each pair of objects in $\opt(\cat{C})$, in other words, to each $4$-tuple of objects of $\cat{C}$, such that this assignment is natural in $\opt(\cat{C})$.
	Taking $P_1=\suphom{a}{a'}$, $P_2=\suphom{b}{b'}$ and $R=\suphom{c}{c'}$ and unpacking the naturality square with respect to functoriality in $\opt(\cat{C})$ encodes the following condition
	\begin{equation*}
  		\tikzfigscale{1}{figs/multi_htensor_proof_2} = \tikzfigscale{1}{figs/multi_htensor_proof_1}  
	\end{equation*}
	where we have informally drawn wires running over $\eta$ to save vertical space. It is clear that this diagrammatic equation packages both \eqref{eq:multi_nat_law} and \eqref{eq:multi_strength_law}. 
\end{proof}

\subsection{Proof of Lemma \ref{lem:seq}}\label{sec:seq_proof}
\begin{lemma*}(Restated for clarity).
	On any symmetric monoidal category $\cat{C}$ the morphisms of type \[ \eta : \suphom{b}{b'} \seq \suphom{a}{a'} \rightarrow\suphom{c}{c'} \] in $\StProf(\cat{C})$ are the families of functions of type \[\eta_{x,x',z}: \cat{C}(a \otimes x, a' \otimes z)\times \cat{C}(b \otimes z , b' \otimes x') \rightarrow \cat{C}(c \otimes x, c' \otimes x')\] satisfying the equations \eqref{eq:vert_naturality} and \eqref{eq:vert_dinaturality}.
\end{lemma*}
\begin{proof}
	Let $P_1$, $P_2$ and $Q$ be strong endoprofunctors on $\cat{C}$ and consider a strong natural transformation $\eta:P_1\seq P_2= \int^z P_1(-,z)\times P_2(z,\bl) \morph{}R(-,\bl)$.
	Upon pre-composing the components of $\eta$ with the coprojections into the coend, we see that $\eta$ is equivalently a family of functions $\eta_{x,x',z}: P_1(x,z) \times P_2(z,x') \rightarrow Q(x,x')$ which are natural in $(x,x')$ in $\opt(\cat{C})$ and dinatural in $z$ in $\cat{C}$.
	Picking $P_1$, $P_2$ and $R$ to be of the form $\suphom{a}{a'}$, the naturality and dinaturality of $\eta$ are the data given in equations \eqref{eq:vert_naturality} and \eqref{eq:vert_dinaturality} respectively.
\end{proof}

\subsection{Proof of Theorem \ref{thm:seq_supermaps}}\label{sec:seq_supermaps_proof}
\begin{theorem*}(Restated for clarity).
	The quantum supermaps on the $n$-combs are the morphisms of strong profunctors of the following type in $\StProf(\cptp)$.
	\begin{equation}\label{eq:supermap_seq_type_app}
		 \bigseq_{i=1}^n  \cptp(a_i \otimes - ,a_i' \otimes \bl) \rightarrow \cptp(c \otimes - ,c' \otimes \bl)
	\end{equation}
\end{theorem*}
\begin{proof}
	Note that the convention of \cite{chiribella_networks} is adopted, in which the phrase $n$-comb refers to combs with $(n-1)$-inputs, so that a $1$-comb is a channel. It is clear from Lemma \ref{lem:seq} that any quantum supermap on $n$-combs can be used to construct a morphism of the type \eqref{eq:supermap_seq_type_app}, what remains is the converse.

	We focus on the $n = 2$ case, the generalisation to larger $n$ only requires more bookkeeping.
	First, observe that the equivalence between optics and combs of $\cptp$ (Theorem \ref{thm:cptp_equiv}), gives the following string of isomorphisms.
	\begin{align*}
		& \cptp(a_2 \otimes - ,a_2' \otimes \bl)  \seq  \cptp(a_1 \otimes - ,a_1' \otimes \bl) \\
		& \cong \opt(\cptp)((a_1',a_2), (a_1 \otimes -, a_2' \otimes \bl)) \\
		& \cong \comb(\cptp)((a_1',a_2),(a_1 \otimes -, a_2' \otimes \bl)) \\
		& \cong C_2 ((a_1',a_2),(a_1 \otimes -, a_2' \otimes \bl))
	\end{align*}
	$C_2 ((a_1',a_2),(a_1 , a_2' ))$ is defined to be the set of all morphisms of type $a_2 \otimes a_1 \rightarrow a_1 '  \otimes a_2 ' $ given by insertion of the swap into a comb of type $\comb(\cptp)((a_1',a_2),(a_1 \otimes -, a_2' \otimes \bl)) $.
	(That this last step is an isomorphism follows from the fact that $\cptp$ embeds into the compact closed category $\mathsf{CP}$ of all CP channels).
	By the basic circuit decomposition theorem for supermaps \cite{chiribella_supermaps}, $C_2$ is equivalently the set of all channels such that for every channel $C_1$,
	\begin{equation}
		\tikzfigscale{1}{figs/definite_proof_1} \in \cptp(a_1,a_2 ')
	\end{equation}
	where this is a diagram in $\mathsf{CP}$.

	Now consider the set $\textrm{dext}_{x,x'}(M)$ of dilation extensions of a set $M \subseteq \cptp(m,m')$ \cite{wilson_locality}.
	\begin{align*} 
		\textrm{dext}_{x,x'}(M) : = \{ &  \phi \in \cptp(m\otimes x,m'\otimes x')  \textrm{ s.t. } \forall \rho : i \rightarrow x,   \\
		& (1 \otimes Tr_{x'}) \circ \phi \circ (1 \otimes \rho) \in M  \} .  
	\end{align*} 
	These dilation extensions can be easily verified to be strong profunctors, functorial in $(x,x')$.

	We will now show that there is an isomorphism of strong profunctors \begin{equation*}
		C_2 ((a_1',a_2),(a_1 \otimes -, a_2' \otimes \bl)) \cong \textrm{dext}_{-,\bl}(C_2 ((a_1',a_2),(a_1 , a_2' )))
	\end{equation*}
	Indeed, see that
	\begin{align*}
		& \forall \rho \ \text{causal}: \quad	\tikzfigscale{1}{figs/definite_proof_2} \in \cptp(a_1,a_2 ') \\
		\implies & \forall \rho, \sigma \ \text{causal}: \quad	\tikzfigscale{1}{figs/definite_proof_3} = 1 \\
		\implies & \tikzfigscale{1}{figs/definite_proof_4} =  \tikzfigscale{1}{figs/definite_proof_5} \\ 
		\implies & \tikzfigscale{1}{figs/definite_proof_6} \in \cptp(a_1 \otimes x,a_2 ' \otimes x')  \\ 
	\end{align*}
	Consequently, we have that:
	\begin{multline*}
		\StProf(\cptp)( \opt(\cptp)((a_1',a_2), (a_1 \otimes -, a_2' \otimes \bl)), \\
		\cptp(c \otimes - , c' \otimes \bl))
	\end{multline*}
	\begin{multline*}
		\cong \StProf(\cptp) (\textrm{dext}_{-,\bl}(C_2 ((a_1',a_2),(a_1 , a_2' ))) , \\
		\cptp(c \otimes - , c' \otimes \bl))
	\end{multline*}
	which by the main theorem of \cite{wilson_locality} is the set of supermaps on the space $C_2 ((a_1',a_2),(a_1 , a_2' ))$ of $2$-combs.
 
	The proof generalises to the $n$-input setting using the following facts.
	\begin{enumerate}
		\item The $n$-arity $\seq$ tensor product gives $n-1$ input optics,
		\item for $\cptp$ those optics are isomorphic to $n-1$ input combs ($n$-combs) by Theorem \ref{thm:cptp_equiv2},
		\item those combs are equivalent (through the equivalency of networks and axiomatic combs of \cite{chiribella_networks}) to supermaps on $n-2$ input combs,
		\item by identical string diagrammatic manipulations the concrete extensions are equivalent to dilation extensions,
		\item by \cite{wilson_locality} the strong morphisms between dilations extensions internalise as standard definition supermaps on the set being dilation extended, completing the proof.
	\end{enumerate}
\end{proof}
	
\subsection{Proof of Theorem \ref{thm:cptp_decomp}}\label{sec:cptp_decomp_proof}
\begin{theorem*}(Restated for clarity).
	The category $\cptp$ has an $n$-arity supermap decomposition theorem for every $n$. 
\end{theorem*}
\begin{proof}
	Beginning with the $1$-arity case, we know that 
	\begin{equation*}
		\StProf(\cptp)\big(\suphom{a}{a'},\suphom{b}{b'}\big)
	\end{equation*}
	is the space of locally applicable transformations $\vb*{a}\morph{}\vb*{b}$.
	By the characterisation theorem of \cite{wilson_locality} these are equivalent to the single-party quantum supermaps which by the decomposition theorem of \cite{chiribella_supermaps} are combs.
	Finally one notes that by Theorem \ref{thm:cptp_equiv} combs of $\cptp$ are equivalent to optics of $\cptp$.
	All together,
	\begin{align*}
		& \StProf(\cptp)\big(\suphom{a}{a'},\suphom{b}{b'}\big) \\
		& \cong \mathsf{lot}(\cat{C})((a, a'),  (b, b')) \\
		& \cong \comb(\cptp)((a , a') , (b , b')) \\
		& \cong \opt(\cptp)((a , a') , (b , b')) \\
		& \cong \StProf(\cat{C})\big(y_{b,b'},y_{a,a'} \big) 
	\end{align*}
	The general $n$-arity case works totally analogously.
\end{proof}

\subsection{Proof of Proposition \ref{prop:decomp_thms}}\label{sec:decomp_thms_proof}
\begin{proposition*}(Restated for clarity).
	A symmetric monoidal category $\cat{C}$ has a $1$-arity decomposition theorem if and only if
	\[\suphom{a}{a'}^* \cong y_{\vb*{a}}, \ \ \textrm{or equivalently}, \ \ y_{\vb*{a}}^{**} \cong y_{\vb*{a}}. \] 
	Furthermore, $\cat{C}$ has an $n$-arity supermap decomposition theorem if and only if 
	\begin{equation*}
		{( \seq_i {y_{\vb*{a}_i }^{*} })}^{*} \cong \seq_i y_{\vb*{a}_i}. 
	\end{equation*}
\end{proposition*}
\begin{proof}
We again begin with the $1$-arity case:
	\begin{align*} 
		&\StProf(\cat{C})\big(\suphom{a}{a'},\suphom{b}{b'}\big) \\
		&\cong \StProf(\cat{C})\big(\suphom{a}{a'},y_{\vb*{b}}^*\big) \\
		&\cong \StProf(\cat{C})\big(y_{\vb*{b}}, \suphom{a}{a'}^*)\\
		&\cong \suphom{a}{a'}^*(\vb*{b}) \cong y_{\vb*{a}}^{\vb*{b}}
	\end{align*} 
	The $n$-arity case is proven identically. 
\end{proof}

\end{document}

%% file: preamble.tex
\usepackage{physics}
\usepackage{stmaryrd}
\usepackage{tikz}
\usepackage{tikz-cd}
\usepackage{tikzit}
\usepackage{circuitikz}
\usepackage{tubes}
\usepackage{xsavebox}

\input{styles.tikzdefs}
\input{styles.tikzstyles}

\newcommand{\tikzfigscale}[2]{\scalebox{#1}{\input{#2.tikz}}}

\newcommand{\morph}[1]{\xrightarrow{#1}}

\newcommand{\pmorph}{\relbar\joinrel\mapstochar\joinrel\rightarrow}
\newcommand{\cat}[1]{\mathcal{#1}}

\newcommand{\opcat}[1]{#1^\textrm{op}}
\newcommand{\set}{\mathsf{Set}}
\newcommand{\cosmos}{\set}

\newcommand{\Cat}{\mathsf{Cat}}
\newcommand{\Prof}{\mathsf{Prof}}
\newcommand{\StProf}{\mathsf{StProf}}

\newcommand{\psh}[1]{[#1^\text{op},\cosmos]}
\newcommand{\copsh}[1]{[#1,\cosmos]}
\newcommand{\opt}{\mathsf{Optic}}
\newcommand{\comb}{\mathsf{Comb}}
\newcommand{\cptp}{\mathsf{CPTP}}
\newcommand{\caus}{\mathsf{Caus}}

\newcommand{\bl}{\mathord{=}}
\newcommand{\suphom}[2]{\cat{C}(#1 \otimes-, #2 \otimes\bl)}
\newcommand{\suphomrev}[2]{\cat{C}(- \otimes #1, \bl \otimes #2)}
\newcommand{\amp}{\mathrel{\rotatebox[origin=c]{180}{\&}}}
\newcommand{\seq}{\varogreaterthan}
\DeclareMathOperator*{\modtensor}{\scalebox{1.2}{$\bigotimes$}_\cat{C}}
\DeclareMathOperator*{\bigseq}{\scalebox{2}{$\ogreaterthan$}}

\DeclareFontFamily{U}{min}{}
\DeclareFontShape{U}{min}{m}{n}{<-> udmj30}{}

\makeatletter
\newcommand{\xRightarrow}[2][]{\ext@arrow 0359\Rightarrowfill@{#1}{#2}}
\makeatother

%% file: styles.tikzstyles
\tikzstyle{discarding}=[fill=white, draw=black, shape=circle, style=upground]
\tikzstyle{smalldiscarding}=[fill=white, draw=black, style=upground, scale=0.75]
\tikzstyle{backdiscard}=[fill=white, draw=black, shape=circle, style=downground, scale=0.5]
\tikzstyle{smallbackdiscard}=[fill=white, draw=black, shape=circle, style=downground, scale=0.5]
\tikzstyle{state}=[fill=white, draw=black, style=triang, tikzit shape=rectangle]
\tikzstyle{kstate}=[fill=white, draw=black, style=kpoint, tikzit shape=rectangle]
\tikzstyle{kstateconj}=[fill=white, draw=black, style=kpoint conjugate, tikzit shape=rectangle]
\tikzstyle{kstateBIG}=[fill=white, draw=black, style=big kpoint, tikzit shape=rectangle]
\tikzstyle{effect}=[fill=white, draw=black, style=triangdag]
\tikzstyle{keffect}=[fill=white, draw=black, style=kpoint adjoint]
\tikzstyle{keffectconj}=[fill=white, draw=black, style=kpoint transpose]
\tikzstyle{morphdag}=[style=mapdag]
\tikzstyle{morph}=[style=hadamard]
\tikzstyle{WIDEmorph}=[style=hadamard, minimum width=14mm]
\tikzstyle{morphtrans}=[style=maptrans]
\tikzstyle{morphconj}=[style=mapconj]
\tikzstyle{CPMmorph}=[style=dmap]
\tikzstyle{CPMmorphconj}=[style=dmapconj]
\tikzstyle{CPMmorphdag}=[style=dmapdag]
\tikzstyle{CPMmorphtrans}=[style=dmaptrans]
\tikzstyle{CPMstate}=[fill=white, draw=black, style=triang, doubled]
\tikzstyle{CPMstateBIG}=[fill=white, draw=black, style={triang_lesssep}, doubled]
\tikzstyle{CPMkstate}=[fill=white, draw=black, style=kpoint, tikzit shape=rectangle, doubled]
\tikzstyle{CPMkstateconj}=[fill=white, draw=black, style=kpoint conjugate, tikzit shape=rectangle, doubled]
\tikzstyle{CPMkstateBIG}=[fill=white, draw=black, style=big kpoint, tikzit shape=rectangle, doubled]
\tikzstyle{CPMkeffect}=[fill=white, draw=black, style=kpoint adjoint, doubled]
\tikzstyle{CPMkeffectconj}=[fill=white, draw=black, style=kpoint transpose, doubled]
\tikzstyle{UHfB}=[fill=white, draw=black, style=triangdag, doubled, inner sep=-2pt]
\tikzstyle{leak}=[style=tinypoint, regular polygon rotate=-90]
\tikzstyle{leakfill}=[style=tinypoint, regular polygon rotate=-90, fill=black]
\tikzstyle{Z}=[style=dot, fill=green]
\tikzstyle{X}=[style=dot, fill=red]
\tikzstyle{black_dot}=[style=dot, fill=black]
\tikzstyle{white_dot}=[style=dot, fill=white]
\tikzstyle{qblack_dot}=[style=ddot, fill=black]
\tikzstyle{qwhite_dot}=[style=ddot, fill=white]
\tikzstyle{whitephase}=[style=wphase dot, fill=white]
\tikzstyle{qredphase}=[style=phase dot, fill=red]
\tikzstyle{qgreenphase}=[style=phase dot, fill=green]
\tikzstyle{had}=[style=hadamard, doubled]
\tikzstyle{box}=[style=hadamard]
\tikzstyle{bigbox}=[style=hadamard, minimum height=4mm, minimum width=8mm]
\tikzstyle{classhad}=[style=hadamard]
\tikzstyle{antipode}=[style=anti]

\tikzstyle{dottededge}=[-, dash pattern=on 1pt off 0.7pt]
\tikzstyle{double edge}=[-, style=doubled, draw=black, tikzit draw={rgb,255: red,18; green,168; blue,191}]
\tikzstyle{new edge style 0}=[<-]
\tikzstyle{new edge style 1}=[-, draw={rgb,255: red,242; green,233; blue,206}, fill={rgb,255: red,242; green,233; blue,206}]
\tikzstyle{morphism_shade}=[-, draw=black, fill={rgb,255: red,242; green,233; blue,206}, line join=bevel]
\tikzstyle{supermap_shade}=[-, fill={rgb,255: red,216; green,215; blue,242}, draw=black, line join=bevel]
\tikzstyle{hole_shade}=[-, fill=white, draw=black,line join=bevel]
\tikzstyle{new edge style 2}=[-, draw={rgb,255: red,14; green,188; blue,83}]
\tikzstyle{new edge style 3}=[<-, draw={rgb,255: red,234; green,209; blue,255}]
\tikzstyle{new edge style 4}=[<-, draw={rgb,255: red,0; green,128; blue,128}]
\tikzstyle{new edge style 5}=[-, draw={rgb,255: red,214; green,110; blue,62}]
\tikzstyle{new edge style 6}=[-, draw={rgb,255: red,174; green,20; blue,174}]